\documentclass[12pt]{article}
\pdfoutput=1

\usepackage{tikz}
\usetikzlibrary{quantikz}
\usetikzlibrary{external}
\usepackage{graphicx}
\usepackage[margin=1in]{geometry}
\usepackage{doi}
\usepackage{hyperref}

\usepackage[utf8]{inputenc}
\usepackage[T1]{fontenc}
\usepackage{amsmath,amsthm,amssymb}
\usepackage{mathtools}

\newcommand{\Id}{\mathrm{Id}}
\newcommand{\SCS}{\mathit{SCS}}
\newcommand{\CNOT}{\mathit{CNOT}}
\newcommand{\wt}{\mathrm{wt}}
\newcommand{\bigO}{\mathcal{O}}
\newcommand{\dicke}[2]{\ket{\smash{D_{#2}^{#1}}}}

\newtheorem{lemma}{Lemma}
\newtheorem{theorem}{Theorem}
\newtheorem{definition}{Definition}

\title{Deterministic Preparation of Dicke States}
\author{Andreas Bärtschi\footnote{Corresponding author: baertschi@lanl.gov\hfill LA-UR-19-22718} \and Stephan Eidenbenz}
\date{Los Alamos National Laboratory}

\begin{document}
\maketitle

\begin{abstract}
	The Dicke state $\dicke{n}{k}$ is an
	equal-weight superposition of all $n$-qubit states with Hamming Weight $k$ (i.e.~all strings of length $n$ with exactly $k$ ones over a binary alphabet).
	Dicke states are an important class of entangled quantum states that among other things serve as starting states for combinatorial optimization quantum algorithms.

	We present a deterministic quantum algorithm for the preparation of Dicke states.
	Implemented as a quantum circuit, our scheme uses $\bigO(kn)$ gates, has depth $\bigO(n)$ and needs no ancilla qubits.	
	The inductive nature of our approach allows for linear-depth preparation of arbitrary symmetric pure states
	and -- used in reverse -- yields a quasilinear-depth circuit for efficient compression of quantum information 
	in the form of symmetric pure states, improving on existing work requiring quadratic depth.
	All of these properties even hold for Linear Nearest Neighbor architectures.	
\end{abstract}

\section{Introduction}
\label{sec:intro}
Within quantum computing, the seemingly mundane task of (efficient) state preparation is actually a separate research topic. 
Recall that a quantum state over $n$ qubits is a superposition $\sum_{x \in \{0.1\}^n } c_x \ket{x}$ of all $2^n$ binary strings $x$  of length $n$ with complex weights $c_x$ such that $\sum_{x \in \{0,1\}^n } |c_x|^2 = 1$. 
The problem of preparing an arbitrary quantum state can be solved with $\Theta(2^n)$ quantum gates~\cite{shende2006synthesis}, 
which can be improved to a polynomial number of gates for states which have a polynomial number of non-zero weights $c_x$.
The intriguing algorithmic question then becomes for what other classes of quantum states do polynomial-time state preparation algorithms exist? 
Very few results exist on this topic and we are still far from having a comprehensive solution, however, Dicke states form such a class: 
the Dicke state $\dicke{n}{k}$ has $\tbinom{n}{k}$ non-zero weights, which is not polynomial in $n$ for super-constant $k$.

Among different types of highly entangled states, the family of Dicke states~\cite{Dicke1954} has garnered widespread attention for tasks in quantum networking~\cite{Prevedel2009}, quantum game theory~\cite{Oezdemir2007}, quantum metrology~\cite{Toth2012} 
and as starting states for combinatorial optimization problems via adiabatic evolution~\cite{Childs2002}. 
Perhaps most promisingly -- Dicke states can be used in the Quantum Alternating Operator Ansatz (QAOA) framework~\cite{Farhi2014,Hadfield2019} for combinatorial optimization problems with hard constraints, 
as a starting state for the actual QAOA algorithm where they represent a superposition of all feasible solutions (in some problem variations).

\newpage
\begin{definition}\label{def:dicke-state}
	A Dicke state $\dicke{n}{k}$\footnote{Various symbols for Dicke states are used in the existing literature, e.g., $\ket{\smash{D^n_k}}$, $\ket{\smash{D_n^{(k)}}}$, $\ket{\genfrac{}{}{0pt}{}{n}{k}}$ or $\ket{n;k}$.}
	is the equal superposition of all $n$-qubit states $\ket{x}$ with Hamming weight $\wt(x)=k$,
	\[	\dicke{n}{k} = \tbinom{n}{k}^{-\frac{1}{2}} \sum\nolimits_{x \in \left\{ 0,1 \right\}^n,\ \mathrm{wt}(x)=k}{\ket{x}}.	\]
\end{definition}

We have, e.g., $\dicke{4}{2} = \tfrac{1}{\sqrt{6}}\left(\ket{1100}+\ket{1010}+\ket{1001}+\ket{0110}+\ket{0101}+\ket{0011}\right)$,
a state that has been studied for its entanglement properties: from $\dicke{4}{2}$, we can generate both 3-qubit $W_3$-states $\dicke{3}{1}$ and $\mathit{GHZ}$-class $G_3$-states $\tfrac{1}{\sqrt{2}} (\dicke{3}{1}-\dicke{3}{2})$ 
by a (local) projective measurement of the same qubit~\cite{Kiesel2007}, whereas these two states cannot be transformed into each other by stochastic local manipulations~\cite{Duer2000}; 
these types of basic transformations of states are non-trivial in quantum computing.

\paragraph{Result Overview.} 
Despite successful experimental creation of Dicke states in physical systems such as trapped ions~\cite{Hume2009,Ivanov2013,Lamata2013}, atoms~\cite{Stockton2004,Xiao2007,Shao2010}, photons~\cite{Prevedel2009,Wieczorek2009} and superconducting qubits~\cite{Wu2016},
efficient quantum \emph{circuits} for the preparation of arbitrary Dicke states $\dicke{n}{k}$ have received little attention.
In this paper, we present -- as our main contribution -- a circuit for deterministic preparation of Dicke states which, given as input the easily prepared classical state $\ket{0}^{\otimes n-k}\ket{1}^{\otimes k}$, prepares the Dicke state $\dicke{n}{k}$.
Our circuit has depth $\bigO(n)$ -- independent of $k$ -- and needs $\bigO(kn)$ gates in total. 
Circuit depth is equivalent to run time and gate count is a measure for overall resource needs. 
In fact, any difference between gate count and depth can be attributed to gate-level parallelism. 
Finding minimal-depth circuits is particularly crucial for Noisy Intermediate Scale Quantum (NISQ) devices, which do not allow for full error correction, and thus experience (unwanted) decoherence the longer a computation lasts. 
Minimizing overall gate count is crucial as each gate operation introduces noise, thus impacting result quality.

Leveraging our main result, we show 
(i) that all our bounds even hold for Linear Nearest Neighbor architectures, where each qubit is connected only to its two neighbors, which is a more realistic assumption for most NISQ devices than the standard all-to-all connectivity, 
(ii) that our circuit can be extended to prepare arbitrary symmetric pure states using linear circuit depth, where a state is symmetric if it is invariant under permutation of the qubits, and
(iii) how to use our construction for compression of quantum information, which is the problem of compressing a symmetric pure $n$-qubit state into $\lceil \log(n+1) \rceil$ qubits without information loss.

\paragraph{Previous Approaches.} 
Previous work has prepared Dicke states probabilistically with success probability $\Omega(\tfrac{1}{\sqrt{n}})$ by applying a biased Hadamard transform to each qubit~\cite{Childs2002},
followed by postselecting the Dicke state through addition of each of the $n$ qubits into an ancilla register of size $\log n$ initialized to the $\ket{0}$ state~\cite{Chuang2000} and a projective measurement thereof. 
A later contribution~\cite{Chakraborty2014} uses a more involved preparation strategy -- giving numerical evidence of a constant-factor improved probability -- followed by a generalized parity measurement~\cite{Ionicioiu2008}, 
also pointing out the potential use of amplitude amplification.
Deterministic preparation circuits without the use of ancilla qubits have been known for the special case of $W$-states $\dicke{n}{1}$, 
either by an iterative construction of quadratic circuit size and depth~\cite{Diker2016} or by a linear number of large multi-controlled rotation gates~\cite{MSQuantumKatas}.
An inductive approach to construct Dicke states up to error $\varepsilon$~\cite{Mosca2001} uses $\Omega(\log k + \log\tfrac{1}{\varepsilon})$ ancilla qubits to count the Hamming weight of the qubits processed so far, 
to then use this register as a control for rotation gates on the next qubit, yielding a superlinear circuit size and depth overall. 
Our approach improves on all of these results in terms of circuit size and depth; additionally, it does not require ancilla qubits, is fully deterministic and in some cases more general.

\paragraph{Relation to Quantum Compression.} 
There exists an interesting relationship between Dicke states and quantum compression. Quantum compression can be understood through the quantum Schur-Weyl transform~\cite{Bacon2006}, 
which separates the angular momentum information of a state from its -- for symmetric states trivial -- permutation information.

The Schur-Weyl transform has been implemented experimentally for a separable symmetric $3$-qubit state~\cite{Rozema2014}, i.e. a state of the form
$(\alpha\ket{0} + \beta\ket{1})^{\otimes 3} = \sum_{\ell=0}^3 \alpha^{3-\ell} \beta^{\ell} \smash{\binom{3}{\ell}}^{1/2} \dicke{3}{\ell}$.
A high-level description of a circuit for general $n$, using no ancilla qubits, has also been developed~\cite{Plesch2010}.
The major circuit part in~\cite{Plesch2010} is of size and depth $\Theta(n^2)$ and maps each Dicke state $\dicke{n}{\ell}$ to the state $\ket{0}^{\otimes \ell-1}\ket{1}\ket{0}^{\otimes n-\ell}$, too.
Its inverse circuit can therefore be used to prepare Dicke states with depth $\Theta(n^2)$.
Our approach in reverse, on the other hand, will yield a quantum compression circuit of size $\bigO(n^2)$ and reduced quasilinear depth $\tilde{\bigO}(n)$, 
where $\tilde{\bigO}(\cdot)$ hides polylogarithmic factors due to the compression part of the circuit 
(mapping terms of the form $\ket{0}^{\otimes \ell-1}\ket{1}\ket{0}^{\otimes n-\ell}$ or $\ket{0}^{\otimes n-\ell}\ket{1}^{\otimes \ell}$, respectively, into $\lceil \log(n+1)\rceil$ qubits).

\paragraph{Outline.} 
This article is organized as follow: In Section~\ref{sec:preparation}, we present an iterative construction of a circuit for deterministic preparation of arbitrary Dicke states. 
We analyze its gate count and circuit depth, and extend these bound to Linear Nearest Neighbor architectures in Section~\ref{sec:circuit-size}.
In Section~\ref{sec:superposition-compression}, we show how our construction can be used to create arbitrary symmetric pure states, written as a superposition of Dicke states, and we present an improved scheme for efficient compression of quantum information.
A detailed comparison with existing work on quantum compression as well as a step-by-step example is in Appendix~\ref{app:comparison}.

\section{Deterministic Dicke State Preparation}
\label{sec:preparation}

In order to prepare Dicke states, we design a unitary operator $U_{n,k}$ which, given as input the classical state $\ket{0}^{\otimes n-k}\ket{1}^{\otimes k}$ (which appears itself as a term in the superposition $\dicke{n}{k}$), 
generates the entire Dicke state $\dicke{n}{k}$. Additionally,  $U_{n,k}$ also generates Dicke states $\dicke{n}{\ell}$ for smaller~$\ell < k$, 
when given as input a string $\ket{0}^{\otimes n-\ell}\ket{1}^{\otimes \ell}$:
\begin{definition}
	Denote by $U_{n,k}$ any unitary satisfying for all $\ell \leq k\colon U_{n,k}\ket{0}^{\otimes n-\ell}\ket{1}^{\otimes \ell} = \dicke{n}{\ell}$.%
	\label{def:unitary}
\end{definition}
Having this property not only for $\ell =k$ but for all $\ell \leq k$ will allow us to build a unitary $U_{n,k}$ inductively, by making use of the following 
composition (also observed, e.g., in~\cite{Lamata2013,Moreno2018}):
\begin{lemma}
	\label{lem:sum-form}
	Dicke states $\dicke{n}{\ell}$ have the inductive sum form
	\[	\dicke{n}{\ell} = \sqrt{\tfrac{\ell}{n}} \dicke{n-1}{\ell-1} \otimes\ket{1} + \sqrt{\tfrac{n-\ell}{n}} \dicke{n-1}{\ell} \otimes\ket{0}.	\]
\end{lemma}
\begin{proof}
	We simply rewrite $\dicke{n}{\ell} := \binom{n}{\ell}^{-\frac{1}{2}} \sum\nolimits_{x \in \left\{ 0,1 \right\}^n,\, \mathrm{wt}(x)=\ell}{\ket{x}}$ as
	\begin{align*}
		\dicke{n}{\ell}	& = \sqrt{\tfrac{1}{\binom{n}{\ell}}} \sum_{\substack{x\in\left\{ 0,1 \right\}^{n-1}\\ \mathrm{wt}(x)=\ell-1}}{\ket{x}}\otimes \ket{1}
				  + \sqrt{\tfrac{1}{\binom{n}{\ell}}} \sum_{\substack{x\in\left\{ 0,1 \right\}^{n-1}\\ \mathrm{wt}(x)=\ell}}{\ket{x}}\otimes \ket{0}	\\
				& = \sqrt{\tfrac{\binom{n-1}{\ell-1}}{\binom{n}{\ell}}} \dicke{n-1}{\ell-1} \otimes\ket{1} 
				  + \sqrt{\tfrac{\binom{n-1}{\ell}}{\binom{n}{\ell}}} \dicke{n-1}{\ell} \otimes\ket{0}							\\
				& = \sqrt{\tfrac{\ell}{n}} \dicke{n-1}{\ell-1} \otimes\ket{1} 
				  + \sqrt{\tfrac{n-\ell}{n}} \dicke{n-1}{\ell} \otimes\ket{0}.\qedhere
	\end{align*}
\end{proof}

The Dicke states $\dicke{n-1}{\ell-1}$ and $\dicke{n-1}{\ell}$ can both be prepared by the same unitary $U_{n-1,k}$ 
given the classical input states $\ket{0}^{\otimes n-\ell}\ket{1}^{\otimes \ell-1}$ and $\ket{0}^{\otimes n-1-\ell}\ket{1}^{\otimes \ell}$, respectively. 
The idea is therefore -- in order to inductively design $U_{n,k}$ -- to apply the composition given by Lemma~\ref{lem:sum-form} to the \emph{input states} $\ket{0}^{\otimes n-\ell}\ket{1}^{\otimes \ell}$ for all $\ell \leq k$, before applying the smaller unitary $U_{n-1,k}$.
Hence for $\ell \leq k$, we are looking for unitary transformations of the form 
\[ \ket{0}^{\otimes n-k-1}\ket{0}^{\otimes k+1-\ell}\ket{1}^{\otimes\ell} \mapsto \sqrt{\tfrac{\ell}{n}} \ket{0}^{\otimes n-k-1}\ket{0}^{\otimes k+1-\ell}\ket{1}^{\otimes\ell} + \sqrt{\tfrac{n-\ell}{n}}  \ket{0}^{\otimes n-k-1}\ket{0}^{\otimes k-\ell}\ket{1}^{\otimes\ell}\ket{0}. \]
Note that this transformation acts trivially on the first $n-k-1$ qubits. Intuitively, it can be described as taking the last $k+1$ of $n$ qubits as an input, splitting the input term into a superposition of two parts, and cyclicly shifting the second part by one position to the left.
We call a unitary that simultaneously implements this transformation for all $\ell \in 0,\dots,k$ a \emph{Split} \& \emph{Cyclic Shift} unitary $\SCS_{n,k}$:

\begin{definition}
	\label{def:scs}
	Denote by $\SCS_{n,k}$ any unitary satisfying for all $\ell \in 1,\dots,k$, where $k<n$:
	\begin{align*}
		\SCS_{n,k} \ket{0}^{\otimes k+1}\phantom{\ket{1}^{\otimes\ell-\ell}}	& =	\ket{0}^{\otimes k+1}, 			\\
		\SCS_{n,k} \ket{0}^{\otimes k+1-\ell}\ket{1}^{\otimes\ell}		& = \sqrt{\tfrac{\ell}{n}} \ket{0}^{\otimes k+1-\ell}\ket{1}^{\otimes\ell} + \sqrt{\tfrac{n-\ell}{n}}  \ket{0}^{\otimes k-\ell}\ket{1}^{\otimes\ell}\ket{0},\\
		\SCS_{n,k} \ket{1}^{\otimes k+1}\phantom{\ket{1}^{\otimes\ell-\ell}}	& =	\ket{1}^{\otimes k+1}.	
	\end{align*}
\end{definition}

Before we describe the inductive construction of the unitaries $U_{n,k}$ in terms of unitaries $\SCS_{n,k}$ and $U_{n-1,k}$, we review $\SCS_{n,k}$ (acting on the last $k+1$ qubits) and $U_{n-1,k}$ (acting on the first $n-1$ qubits) in comparison:
\vspace{-5ex}
\begin{center}
	\begin{minipage}[t]{0.6\textwidth}
		\begin{align*}
			\SCS_{n,k}\colon
			\ket{00..000}	& \mapsto	\ket{00..000}											\\[1ex]
			\ket{00..001}	& \mapsto	\smash{\sqrt{\tfrac{1}{n}}} \ket{00..001} + \smash{\sqrt{\tfrac{n-1}{n}}} \ket{00..010}		\\[1ex]
			\ket{00..011}	& \mapsto	\smash{\sqrt{\tfrac{2}{n}}} \ket{00..011} + \smash{\sqrt{\tfrac{n-2}{n}}} \ket{00..110}		\\[-6pt]
			\		& \hspace*{6pt}\vdots												\\[-4pt]
			\ket{01..111}	& \mapsto	\smash{\sqrt{\tfrac{k}{n}}} \ket{01..111} + \smash{\sqrt{\tfrac{n-k}{n}}} \ket{11..110}		\\[1ex]
			|\hspace*{-1pt}\underbrace{\hspace*{-1pt}11..111\hspace*{-1pt}}_{k+1}\hspace*{1pt}\rangle	& \mapsto \ket{11..111}
		\end{align*}
	\end{minipage}%
	\begin{minipage}[t]{0.4\textwidth}
		\begin{align*}
			U_{n-1,k}\colon
			\ket{0..000..00}	& \mapsto	\dicke{n-1}{0}	\\[1ex]
			\ket{0..000..01}	& \mapsto	\dicke{n-1}{1}	\\[1ex]
			\ket{0..000..11}	& \mapsto	\dicke{n-1}{2}	\\[-6pt]
			\		& \hspace*{6pt}\vdots			\\[-4pt]
			\ket{0..001..11}	& \mapsto	\dicke{n-1}{k-1}\\[1ex]
			|\hspace*{-6pt}\underbrace{0..0}_{n-1-k}\hspace*{-6pt}\underbrace{\hspace*{-1pt}11..11\hspace*{-1pt}}_{k}\hspace*{1pt}\rangle	& \mapsto \dicke{n-1}{k}
		\end{align*}
	\end{minipage}
\end{center}

\subsection[Inductive Construction of U\_n,k]{Inductive Construction of $U_{n,k}$}
An explicit construction of Split \& Cyclic Shift unitaries in terms of standard gates will be given in Subsection~\ref{sec:preparation:construction}.
For now, however, we will show how arbitrary $U_{n,k}$ unitaries can be constructed inductively from unitaries $\SCS_{n,k}$ (acting on the last $k+1$ qubits) and $U_{n-1,k}$ (acting on the first $n-1$ qubits).
Clearly, we must have $U_{1,1} = \Id$. We construct unitaries of the form $U_{k,k}$ by iteratively applying $\SCS_{k,k-1}$ immediately before $U_{k-1,k-1}$, i.e. $U_{k,k} = (U_{k-1,k-1} \otimes \Id) \cdot \SCS_{k,k-1}$.
Arbitrary unitaries $U_{n,k}$ can be built by preceding $U_{n-1,k}$ with $\SCS_{n,k}$, as shown in Figure~\ref{fig:induction}, giving 
$U_{n,k} = (U_{n-1,k} \otimes \Id) \cdot (\Id^{\otimes n-k-1} \otimes \SCS_{n,k})$.\footnote{An inductive approach
which ``sandwiches'' smaller unitaries has previously been used for $W$-states~$\dicke{n}{1}$, albeit with depth $\bigO(n^2)$~\cite{Diker2016}.} 
Telescoping these recursions we get:
\begin{lemma}
	\label{lem:induction}
	The following inductive construction of $U_{n,k}$ is consistent with Definition~\ref{def:unitary}:
	\[ U_{n,k}	 
			:=	\prod_{\ell=2}^{k} \left( \SCS_{\ell,\ell-1} \otimes \Id^{\otimes n-\ell} \right) \cdot \prod_{\ell=k+1}^{n} \left( \Id^{\otimes \ell-k-1}\otimes \SCS_{\ell,k} \otimes \Id^{\otimes n-\ell} \right).	\]%
\end{lemma}

\begin{figure}[t!]
	\centering
	\begin{tikzpicture}
	\clip (-8.1,1.6) rectangle (8.4,-1.6);
	\node[scale=0.645]{
	\begin{tikzcd}[row sep={20pt,between origins}, execute at end picture={
			\draw[thick,inner sep=5pt] ($(\tikzcdmatrixname-2-12)+(0pt,5pt)$) rectangle ++(50pt,-70pt) node[midway]       {$\SCS_{n,k}$};
			\draw[thick,inner sep=5pt] ($(\tikzcdmatrixname-2-12)+(60pt,25pt)$) rectangle ++(50pt,-70pt) node[midway] (P) {$\SCS_{n\text{-}1,k}$};
			\draw[thick,inner sep=5pt] ($(\tikzcdmatrixname-1-12)+(140pt,5pt)$) rectangle ++(50pt,-70pt) node[midway] (Q) {$\SCS_{k+1,k}$};
			\draw[thick,inner sep=5pt] ($(\tikzcdmatrixname-1-12)+(220pt,5pt)$) rectangle ++(50pt,-40pt) node[midway] (R) {$\SCS_{3,2}$};
			\draw[thick,inner sep=5pt] ($(\tikzcdmatrixname-1-12)+(280pt,5pt)$) rectangle ++(50pt,-30pt) node[midway]     {$\SCS_{2,1}$};
			\node at ($0.5*(P)+0.5*(Q)$) {$\dots$};
			\node at ($0.5*(Q)+0.5*(R)$) {$\dots$};
    	}]							
		& \qwbundle\gate[wires=5]{\quad U_{n,k}\quad}  	& \qw	&&	&							& \lstick{$\ket{0}^{\otimes n-k-1}$}														& \qwbundle\qw													& \gate[wires=4]{\quad U_{n-1,k}\quad}									& \qw	&&	\\[16pt]
		& 						& \qw	&=\;&	& \lstick[wires=4]{\rotatebox{90}{$k+1$}}\hspace*{44pt}	& \lstick{$n-k\colon \ket{\ \,}$} 														& \gate[wires=4]{\parbox[c][10pt][c]{2.5cm}{\centering Shift \&\\ Cyclic Split\\[4pt] $\SCS_{n,k}$}}	 	& 													& \qw	&=\,&	\\
		&						& \qw	&&	& 							& \lstick{\raisebox{6pt}{\vdots}\hspace*{6pt}}													& 														&													& \qw	&&	\\
		& 						& \qw	&&	&							& \lstick{$n-1\colon \ket{\ \,}$}\slice[label style={inner sep=1pt,anchor=south west,xshift=-2pt,rotate=30}]{\ket{\text{slice}1}}		&														& \slice[label style={inner sep=1pt,anchor=south west,xshift=-2pt,rotate=30}]{\ket{\text{slice}3}}	& \qw	&&	\\ 
		& 						& \qw	&&	&							& \lstick{$n\colon \ket{\ \,}$}															& \slice[label style={inner sep=1pt,anchor=south west,xshift=-2pt,rotate=30}]{\ket{\text{slice}2}}		& \qw													& \qw	&&	
	\end{tikzcd}
	};
	\end{tikzpicture}\\[-2ex]
	\caption{Inductive construction of the unitaries $U_{n,k}$ from the unitaries $\SCS_{n,k}$ and $U_{n-1,k}$.\label{fig:induction}}
\end{figure}
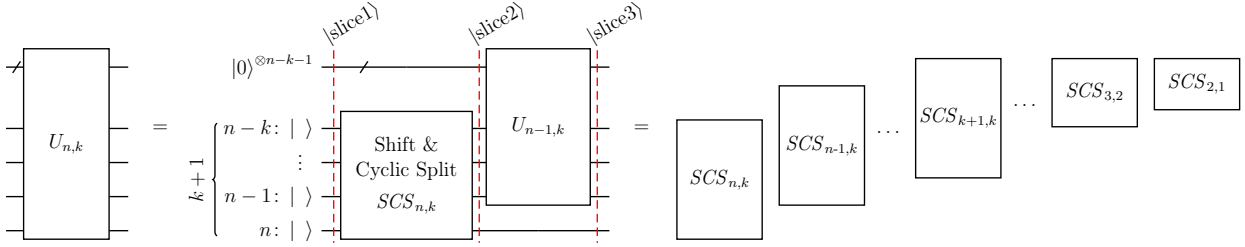

\begin{proof}
	We show by induction over $n$ that for all $\ell \leq k$ we have $U_{n,k} \ket{0}^{\otimes n-\ell}\ket{1}^{\otimes \ell} = \dicke{n}{\ell}$.

	For the base case $U_{2,2}$ we have directly by Definition~\ref{def:scs} that $\SCS_{2,1}\ket{00} = \ket{00} =: \dicke{2}{0}$, $\SCS_{2,1}\ket{01} = \tfrac{1}{2}(\ket{01} + \ket{10}) =: \dicke{2}{1}$, $\SCS_{2,1}\ket{11} = \ket{11} =: \dicke{2}{2}$
	and thus for $\ell \leq 2$ that $U_{2,2}\ket{0}^{\otimes 2-\ell}\ket{1}^{\otimes \ell} := \SCS_{2,1}\ket{0}^{\otimes 2-\ell}\ket{\otimes 1}^{\ell} = \dicke{n}{\ell}$.
	We proceed by induction, first considering the simpler step $U_{n-1,k} \rightarrow U_{n,k}$, where $k \leq n-1$, then moving on to the step $U_{k-1,k-1}\rightarrow U_{k,k}$:

	\subparagraph{Step $U_{n-1,k} \rightarrow U_{n,k}$.}
	We show $U_{n,k} \ket{0}^{\otimes n-\ell}\ket{1}^{\otimes \ell} = \dicke{n}{\ell}$ for all $\ell \leq k \leq n-1$ by analyzing the three time-slices depicted in Figure~\ref{fig:induction}.
	As input to the circuit we have a corresponding state $\ket{\mathrm{slice}1} = \ket{0}^{\otimes n-\ell}\ket{1}^{\otimes \ell}$ for some $\ell \in 0,\ldots,k$. Applying Lemma~\ref{lem:sum-form} at the end, we get 
	\begin{align*}
		\ket{\mathrm{slice}2}
		& = ( \Id^{\otimes n-k-1} \otimes \SCS_{n,k} ) \ket{0}^{\otimes n-\ell}\ket{1}^{\otimes \ell}				\\ 
		& = \ket{0}^{\otimes n-k-1} \otimes \left(	\sqrt{\tfrac{\ell}{n}} \ket{0}^{\otimes k+1-\ell}\ket{1}^{\otimes \ell}	+ \sqrt{\tfrac{n-\ell}{n}} \ket{0}^{\otimes k-\ell}\ket{1}^{\otimes \ell}\ket{0} \right)	\\
		& = \sqrt{\tfrac{\ell}{n}} \ket{0}^{\otimes (n-1)-(\ell-1)}\ket{1}^{\otimes \ell-1} \otimes \ket{1}
		  + \sqrt{\tfrac{n-\ell}{n}} \ket{0}^{\otimes (n-1)-\ell}\ket{1}^{\otimes \ell}\otimes\ket{0},			\\
		\ket{\mathrm{slice}3}
		& = \left( U_{n-1,k} \otimes \Id \right) \ket{\mathrm{slice}2}						
		  = \sqrt{\tfrac{\ell}{n}} \dicke{n-1}{\ell-1} \otimes\ket{1} + \sqrt{\tfrac{n-\ell}{n}} \dicke{n-1}{\ell} \otimes\ket{0}
		  = \dicke{n}{\ell}.
	\end{align*}
	
	\subparagraph{Step $U_{k-1,k-1} \rightarrow U_{k,k}$.}
	We show $U_{k,k} \ket{0}^{\otimes k-\ell}\ket{1}^{\otimes \ell} = \dicke{k}{\ell}$ for all $\ell \leq k$. Replacing $k$ by $k-1$ and $n$ by $k$ in the previous analysis, 
	we immediately get that for $\ell \leq k-1$ the state $\ket{\mathrm{slice}1} = \ket{0}^{\otimes k-\ell}\ket{1}^{\otimes \ell}$ 
	maps to $\ket{\mathrm{slice}3} = \dicke{k}{\ell}$. It remains to show the same for $\ell = k$:
	\begin{align*}
		\ket{\mathrm{slice}3}
		& = \left( U_{k-1,k-1} \otimes \Id \right) \cdot \SCS_{k,k-1} \ket{\mathrm{slice}1} = \left( U_{k-1,k-1} \otimes \Id \right) \cdot \SCS_{k,k-1} \ket{1}^{\otimes k}	\\ 
		& = \left( U_{k-1,k-1} \otimes \Id \right) \ket{1}^{\otimes k} = \dicke{k-1}{k-1} \otimes \ket{1}  = \ket{1}^{\otimes k} = \dicke{k}{k}.\qedhere
	\end{align*}
\end{proof}

\subsection[Explicit Construction of SCS\_n,k]{Explicit Construction of $\SCS_{n,k}$}
\label{sec:preparation:construction}

In the following, we describe a clean construction of an arbitrary Split \& Cyclic Shift unitary $\SCS_{n,k}$ in terms of $1$ two-qubit gate and $k-1$ three-qubit gates, 
each of which implements exactly one of the $k$ non-trivial mappings for $\ell \in 1,\dots,k$ given in Definition~\ref{def:scs}: 
\begin{equation}
	\ket{0}^{\otimes k+1-\ell}\ket{1}^{\otimes \ell} \rightarrow \sqrt{\tfrac{\ell}{n}} \ket{0}^{\otimes k+1-\ell}\ket{1}^{\otimes\ell} + \sqrt{\tfrac{n-\ell}{n}}  \ket{0}^{\otimes k-\ell}\ket{1}^{\otimes\ell}\ket{0}.	\label{eq:croty} 
\end{equation}
\paragraph{Building Blocks.} 
The relevant qubits in this mapping are the last ($n$th) qubit as well as the pair of qubits in which there is a change in the binary string from $0$'s to $1$'s 
(the $(n-\ell)$th and $(n-\ell+1)$th qubits). Using the notation $\ket{xy}_a$ to say that qubits $a$ and $a+1$ are in states $x$ and $y$, respectively, 
the two- and three-qubit gates are defined and easily constructed by:\\
\newcommand{\blocki}{\mathit{(i)}}
\newcommand{\blockii}[1]{\mathit{(ii)}_{#1}}
\begin{minipage}{0.65\linewidth}
	\begin{align*}
		\blocki\quad
		\ket{00}_{n-1}\phantom{\ket{0}_n}	& \rightarrow	\ket{00}_{n-1}								\\
		\ket{11}_{n-1}\phantom{\ket{0}_n}	& \rightarrow	\ket{11}_{n-1}								\\[-1ex]
		\ket{01}_{n-1}\phantom{\ket{0}_n}	& \rightarrow	\sqrt{\tfrac{1}{n}}\ket{01}_{n-1} + \sqrt{\tfrac{n-1}{n}}\ket{10}_{n-1}	\\[3ex]	
		\blockii{\ell}\quad
		\ket{00}_{n-\ell}\ket{0}_n		& \rightarrow	\ket{00}_{n-\ell}\ket{0}_n						\\
		\ket{01}_{n-\ell}\ket{0}_n		& \rightarrow	\ket{01}_{n-\ell}\ket{0}_n						\\
		\ket{00}_{n-\ell}\ket{1}_n		& \rightarrow	\ket{00}_{n-\ell}\ket{1}_n						\\
		\ket{11}_{n-\ell}\ket{1}_n		& \rightarrow	\ket{11}_{n-\ell}\ket{1}_n						\\[-1ex]
		\ket{01}_{n-\ell}\ket{1}_n		& \rightarrow	\sqrt{\tfrac{\ell}{n}} \ket{01}_{n-\ell}\ket{1}_n + \sqrt{\tfrac{n-\ell}{n}} \ket{11}_{n-\ell}\ket{0}_n	
	\end{align*}
\end{minipage}%
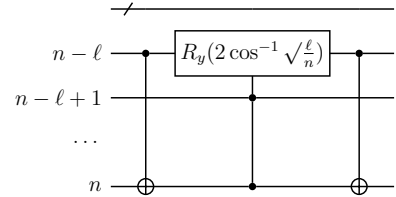
\begin{minipage}{0.35\linewidth}
	\vspace*{2ex}
	\begin{tikzpicture}
	\node[scale=0.7]{
	\begin{tikzcd}[row sep={24pt,between origins}]
		\lstick{\ }	  	& \qw\qwbundle	& \qw							& \qw		& \qw	\\
		\lstick{$n-1$}  	& \ctrl{1}    	& \gate{R_y(2 \cos^{-1} \surd \tfrac{1}{n})}  		& \ctrl{1}	& \qw	\\	
		\lstick{$n$}	 	& \targ{}     	& \ctrl{-1}		  				& \targ{}	& \qw	\\[7ex]	
		\lstick{\ }	  	& \qw\qwbundle	& \qw							& \qw		& \qw	\\
		\lstick{$n-\ell$}  	& \ctrl{3}	& \gate{R_y(2 \cos^{-1} \surd \tfrac{\ell}{n})}		& \ctrl{3}	& \qw	\\
		\lstick{$n-\ell+1$}  	& \qw   	& \ctrl{-1}						& \qw		& \qw	\\
		\lstick{\dots}		&		&							&		& 	\\
		\lstick{$n$}	 	& \targ{}     	& \ctrl{-2}		  				& \targ{}	& \qw	
	\end{tikzcd}
	};
	\end{tikzpicture}
\end{minipage}\\[2ex]
The two-qubit gate $\blocki$ and the $k-1$ three-qubit gates $\blockii{\ell}$ for $2\leq \ell \leq k$ are each constructed by a (two-)controlled $Y$-rotation 
$R_y\left(2\cos^{-1}\surd\tfrac{\ell}{n}\right)$ 
mapping $\ket{0} \rightarrow \surd\tfrac{\ell}{n}\ket{0} + \surd\tfrac{n-\ell}{n}\ket{1}$, conjugated with a $\CNOT$ on the last qubit $n$.
Here, we use $R_y(2\theta) = \left( \begin{smallmatrix} \cos\theta & -\sin\theta \\ \sin\theta & \cos\theta \end{smallmatrix} \right)$.

\paragraph{Putting it all together.}
Note that the states $\ket{0}^{\otimes k+1}$ and $\ket{1}^{\otimes k+1}$ remain unchanged under each of the $k$ gates $\blocki, \blockii{\ell}$. 
Furthermore, for any given $1\leq \ell^* \leq k$, there is exactly one of the $k$ gates $\blocki,\blockii{\ell}$ affecting the state $\ket{0}^{\otimes k+1-\ell^*}\ket{1}^{\otimes \ell^*}$, namely the one with matching $\ell = \ell^*$. 
It maps $\ket{01}_{n-\ell^*}\ket{1}_n \rightarrow \surd\tfrac{\ell^*}{n} \ket{01}_{n-\ell^*}\ket{1}_n + \surd\tfrac{n-\ell^*}{n} \ket{11}_{n-\ell^*}\ket{0}_n$, implementing Equation~\eqref{eq:croty}.

The resulting second term $\ket{0}^{\otimes k-\ell^*}\ket{1}^{\otimes \ell^*}\ket{0}$ remains unaffected by all gates $\blockii{\ell}$ with larger $\ell > \ell^*$.  
Hence we can build a complete $\SCS_{n,k}$ gate starting with the two-qubit gate~$\blocki$ followed the $k-1$ three-qubit gates $\blockii{\ell}$ order by increasing $\ell$. 
For an illustration of $\SCS_{5,3}, \SCS_{4,3}, \SCS_{3,2}$ and $\SCS_{2,1}$ -- together composing $U_{5,3}$ -- see Figure~\ref{fig:U53}. 
The example can also be opened and verified in Quirk~\cite{quirk}
\href{https://algassert.com/quirk#circuit=%7B%22cols%22%3A%5B%5B1%2C1%2C%22X%22%2C%22X%22%2C%22X%22%5D%2C%5B%22Chance%22%2C%22Chance%22%2C%22Chance%22%2C%22Chance%22%2C%22Chance%22%5D%2C%5B%22Chance5%22%5D%2C%5B1%2C1%2C1%2C%22%E2%80%A2%22%2C%22X%22%5D%2C%5B1%2C1%2C1%2C%22~d6b5%22%2C%22%E2%80%A2%22%5D%2C%5B1%2C1%2C1%2C%22%E2%80%A2%22%2C%22X%22%5D%2C%5B%22Chance5%22%5D%2C%5B1%2C1%2C%22%E2%80%A2%22%2C1%2C%22X%22%5D%2C%5B1%2C1%2C%22~f38g%22%2C%22%E2%80%A2%22%2C%22%E2%80%A2%22%5D%2C%5B1%2C1%2C%22%E2%80%A2%22%2C1%2C%22X%22%5D%2C%5B%22Chance5%22%5D%2C%5B1%2C%22%E2%80%A2%22%2C1%2C1%2C%22X%22%5D%2C%5B1%2C%22~dcqq%22%2C%22%E2%80%A2%22%2C1%2C%22%E2%80%A2%22%5D%2C%5B1%2C%22%E2%80%A2%22%2C1%2C1%2C%22X%22%5D%2C%5B%22Chance5%22%5D%2C%5B1%2C1%2C%22%E2%80%A2%22%2C%22X%22%5D%2C%5B1%2C1%2C%22~r6tk%22%2C%22%E2%80%A2%22%5D%2C%5B1%2C1%2C%22%E2%80%A2%22%2C%22X%22%5D%2C%5B%22Chance5%22%5D%2C%5B1%2C%22%E2%80%A2%22%2C1%2C%22X%22%5D%2C%5B1%2C%22~q81a%22%2C%22%E2%80%A2%22%2C%22%E2%80%A2%22%5D%2C%5B1%2C%22%E2%80%A2%22%2C1%2C%22X%22%5D%2C%5B%22Chance5%22%5D%2C%5B%22%E2%80%A2%22%2C1%2C1%2C%22X%22%5D%2C%5B%22~jrlp%22%2C%22%E2%80%A2%22%2C1%2C%22%E2%80%A2%22%5D%2C%5B%22%E2%80%A2%22%2C1%2C1%2C%22X%22%5D%2C%5B%22Chance5%22%5D%2C%5B1%2C%22%E2%80%A2%22%2C%22X%22%5D%2C%5B1%2C%22~tue7%22%2C%22%E2%80%A2%22%5D%2C%5B1%2C%22%E2%80%A2%22%2C%22X%22%5D%2C%5B%22Chance5%22%5D%2C%5B%22%E2%80%A2%22%2C1%2C%22X%22%5D%2C%5B%22~7h3f%22%2C%22%E2%80%A2%22%2C%22%E2%80%A2%22%5D%2C%5B%22%E2%80%A2%22%2C1%2C%22X%22%5D%2C%5B%22Chance5%22%5D%2C%5B%22%E2%80%A2%22%2C%22X%22%5D%2C%5B%22~ok6s%22%2C%22%E2%80%A2%22%5D%2C%5B%22%E2%80%A2%22%2C%22X%22%5D%2C%5B%22Chance5%22%5D%5D%2C%22gates%22%3A%5B%7B%22id%22%3A%22~ok6s%22%2C%22name%22%3A%22%E2%88%9A1%2F2%22%2C%22matrix%22%3A%22%7B%7B%E2%88%9A%C2%BD%2C-%E2%88%9A%C2%BD%7D%2C%7B%E2%88%9A%C2%BD%2C%E2%88%9A%C2%BD%7D%7D%22%7D%2C%7B%22id%22%3A%22~tue7%22%2C%22name%22%3A%22%E2%88%9A1%2F3%22%2C%22matrix%22%3A%22%7B%7B%E2%88%9A%E2%85%93%2C-%E2%88%9A%E2%85%94%7D%2C%7B%E2%88%9A%E2%85%94%2C%E2%88%9A%E2%85%93%7D%7D%22%7D%2C%7B%22id%22%3A%22~7h3f%22%2C%22name%22%3A%22%E2%88%9A2%2F3%22%2C%22matrix%22%3A%22%7B%7B%E2%88%9A%E2%85%94%2C-%E2%88%9A%E2%85%93%7D%2C%7B%E2%88%9A%E2%85%93%2C%E2%88%9A%E2%85%94%7D%7D%22%7D%2C%7B%22id%22%3A%22~r6tk%22%2C%22name%22%3A%22%E2%88%9A1%2F4%22%2C%22matrix%22%3A%22%7B%7B%C2%BD%2C-%E2%88%9A%C2%BE%7D%2C%7B%E2%88%9A%C2%BE%2C%C2%BD%7D%7D%22%7D%2C%7B%22id%22%3A%22~q81a%22%2C%22name%22%3A%22%E2%88%9A2%2F4%22%2C%22matrix%22%3A%22%7B%7B%E2%88%9A%C2%BD%2C-%E2%88%9A%C2%BD%7D%2C%7B%E2%88%9A%C2%BD%2C%E2%88%9A%C2%BD%7D%7D%22%7D%2C%7B%22id%22%3A%22~jrlp%22%2C%22name%22%3A%22%E2%88%9A3%2F4%22%2C%22matrix%22%3A%22%7B%7B%E2%88%9A%C2%BE%2C-%C2%BD%7D%2C%7B%C2%BD%2C%E2%88%9A%C2%BE%7D%7D%22%7D%2C%7B%22id%22%3A%22~d6b5%22%2C%22name%22%3A%22%E2%88%9A1%2F5%22%2C%22matrix%22%3A%22%7B%7B%E2%88%9A%E2%85%95%2C-%E2%88%9A%E2%85%98%7D%2C%7B%E2%88%9A%E2%85%98%2C%E2%88%9A%E2%85%95%7D%7D%22%7D%2C%7B%22id%22%3A%22~f38g%22%2C%22name%22%3A%22%E2%88%9A2%2F5%22%2C%22matrix%22%3A%22%7B%7B%E2%88%9A%E2%85%96%2C-%E2%88%9A%E2%85%97%7D%2C%7B%E2%88%9A%E2%85%97%2C%E2%88%9A%E2%85%96%7D%7D%22%7D%2C%7B%22id%22%3A%22~dcqq%22%2C%22name%22%3A%22%E2%88%9A3%2F5%22%2C%22matrix%22%3A%22%7B%7B%E2%88%9A%E2%85%97%2C-%E2%88%9A%E2%85%96%7D%2C%7B%E2%88%9A%E2%85%96%2C%E2%88%9A%E2%85%97%7D%7D%22%7D%2C%7B%22id%22%3A%22~it96%22%2C%22name%22%3A%22%E2%88%9A4%2F5%22%2C%22matrix%22%3A%22%7B%7B%E2%88%9A%E2%85%98%2C-%E2%88%9A%E2%85%95%7D%2C%7B%E2%88%9A%E2%85%95%2C%E2%88%9A%E2%85%98%7D%7D%22%7D%5D%7D}{\textcolor{red}{following this link}}.

\begin{figure}[t!]
	\centering
	\newcommand{\RG}[2]{\gate{\surd\tfrac{#1}{#2}}}
	\begin{tikzpicture}
	\clip (-8.2,1.4) rectangle (8.3,-1.5);
	\node[scale=0.61]{
		\begin{tikzcd}[row sep={24pt,between origins}, execute at end picture={
				\draw[blue,dotted,thick] ($0.5*(\tikzcdmatrixname-2-4)+0.5*(\tikzcdmatrixname-2-5)+(0pt,12pt)$) -- ($0.5*(\tikzcdmatrixname-5-4)+0.5*(\tikzcdmatrixname-5-5)-(0pt,12pt)$);
				\draw[blue,dotted,thick] ($0.5*(\tikzcdmatrixname-2-7)+0.5*(\tikzcdmatrixname-2-8)+(0pt,12pt)$) -- ($0.5*(\tikzcdmatrixname-5-7)+0.5*(\tikzcdmatrixname-5-8)-(0pt,12pt)$);
				\node[yshift=-15pt] at (\tikzcdmatrixname-5-6) {$\SCS_{5,3}$};
				\node[yshift=-15pt] at (\tikzcdmatrixname-5-15) {$\SCS_{4,3}$};
				\node[yshift=-15pt] at ($0.5*(\tikzcdmatrixname-5-22)+0.5*(\tikzcdmatrixname-5-23)$) {$\SCS_{3,2}$};
				\node[yshift=-15pt] at (\tikzcdmatrixname-5-27) {$\SCS_{2,1}$};
			}]
			\lstick{\ket{0}}\slice{}	& \qw		& \qw		& \qw		& \qw		& \qw		& \qw		& \qw		& \qw		& \qw\slice{}	& \qw		& \qw		& \qw		& \qw		& \qw		& \qw		& \ctrl{3}	& \RG{3}{4}	& \ctrl{3}	& \qw		& \qw		& \qw		& \ctrl{2}	& \RG{2}{3}	& \ctrl{2}	& \ctrl{1}	& \RG{1}{2}	& \ctrl{1}	& \qw\rstick[wires=5]{\ \rotatebox{90}{\large\dicke{5}{3}}}	\\	
			\lstick{\ket{0}}	     	& \qw		& \qw		& \qw		& \qw		& \qw		& \qw		& \ctrl{3}	& \RG{3}{5}	& \ctrl{3}	& \qw		& \qw		& \qw		& \ctrl{2}	& \RG{2}{4}	& \ctrl{2}	& \qw		& \control{}	& \qw\slice{}	& \ctrl{1}	& \RG{1}{3}	& \ctrl{1}	& \qw		& \control{}	& \qw		& \targ{}	& \ctrl{-1}	& \targ{}	& \qw					\\	
			\lstick{\ket{1}}	  	& \qw		& \qw		& \qw		& \ctrl{2}	& \RG{2}{5}	& \ctrl{2}	& \qw		& \control{}	& \qw		& \ctrl{1}	& \RG{1}{4}	& \ctrl{1}	& \qw		& \control{}	& \qw		& \qw		& \qw		& \qw		& \targ{}	& \ctrl{-1}	& \targ{}	& \targ{}	& \ctrl{-2}	& \targ{}	& \qw		& \qw		& \qw\slice{}	& \qw					\\	
			\lstick{\ket{1}}  		& \ctrl{1}	& \RG{1}{5}	& \ctrl{1}	& \qw		& \control{}	& \qw		& \qw		& \qw		& \qw		& \targ{}	& \ctrl{-1}	& \targ{}	& \targ{}	& \ctrl{-2}	& \targ{}	& \targ{}	& \ctrl{-3}	& \targ{}	& \qw		& \qw		& \qw		& \qw		& \qw		& \qw\slice{}	& \qw		& \qw		& \qw		& \qw					\\	
			\lstick{\ket{1}}	 	& \targ{}	& \ctrl{-1}	& \targ{}	& \targ{}	& \ctrl{-2}	& \targ{}	& \targ{}	& \ctrl{-3}	& \targ{}	& \qw		& \qw		& \qw		& \qw		& \qw		& \qw		& \qw		& \qw		& \qw		& \qw		& \qw		& \qw		& \qw		& \qw		& \qw		& \qw		& \qw		& \qw		& \qw		
		\end{tikzcd}
	};
	\end{tikzpicture}\\[-2ex]
	\caption{Preparation of the Dicke state $\dicke{5}{3}$ with $\SCS$ gates implementing a unitary $U_{5,3}$.
	$\surd\tfrac{\ell}{n}$-gates are shorthand for $Y$-Rotations $R_y\left(2\cos^{-1}\surd\tfrac{\ell}{n}\right)$, mapping $\ket{0} \rightarrow \surd\tfrac{\ell}{n}\ket{0} + \surd\tfrac{n-\ell}{n}\ket{1}$.}
	\label{fig:U53}
\end{figure}

\section{Circuit Size and Depth}
\label{sec:circuit-size}

We now analyze the size and depth of our circuit construction and show how to adapt the circuit to be used on Linear Nearest Neighbor (LNN) architectures, 
where 2-qubit gates can only be implemented between neighboring qubits:

\begin{theorem}
	Dicke states $\dicke{n}{k}$ can be prepared with a circuit of size $\bigO(\min(k,n-k)\cdot n)$ and depth $\bigO(n)$,
	even on Linear Nearest Neighbor architectures.
	\label{thm:DS-circuit}
\end{theorem}

\begin{figure}[b!]
	\centering
	\newcommand{\RyT}[2]{\gate{R_y(#1\frac{\theta}{#2})}}
	\begin{tikzpicture}
		\clip (-8.2,0.9) rectangle (8.3,-1.0);
	\node[scale=0.69]{
		\begin{tikzcd}[row sep={20pt,between origins}, slice style=blue, execute at end picture={
			\node[fit=(\tikzcdmatrixname-2-13)(\tikzcdmatrixname-3-14),draw,dashed,inner ysep=4pt,inner xsep=6pt,xshift=2pt] {};
			\draw[dashed]	($(\tikzcdmatrixname-1-26)+(0pt,8pt)$) -- ($(\tikzcdmatrixname-1-25)+(-8pt,8pt)$) -- ($(\tikzcdmatrixname-2-25)-(8pt,6pt)$) -- ($(\tikzcdmatrixname-2-26)-(0pt,6pt)$);
			\draw[red,semithick,->]	(\tikzcdmatrixname-5-10) edge[bend left] (\tikzcdmatrixname-2-11);
			\draw[red,semithick,->]	(\tikzcdmatrixname-5-10) edge[bend left] (\tikzcdmatrixname-2-13);
			\draw[red,semithick,->]	(\tikzcdmatrixname-5-22) edge[bend left] (\tikzcdmatrixname-1-23);
			\draw[red,semithick,->]	(\tikzcdmatrixname-5-22) edge[bend left] (\tikzcdmatrixname-1-25);
			\node at (\tikzcdmatrixname-5-10) {\huge \textcolor{red}{$\times$}};
			\node at (\tikzcdmatrixname-5-22) {\huge \textcolor{red}{$\times$}};
		}]																
		& \qw\slice[style={blue,dotted}]{}	& \qw	  	& \qw			& \qw\slice[style={blue,dotted}]{}	& \ctrl{4}	& \qw	&&& \qw		& \qw\slice[style={blue,dotted}]{}	& \qw		& \qw		& \qw		& \qw		& \qw		& \qw		& \qw		& \qw		& \qw		& \qw		& \qw		& \targ{}\slice[style={blue,dotted}]{}	& \ctrl{4}	& \targ{}	& \qw		\\	
		& \qw					& \ctrl{3}	& \gate{R_y(2\theta)}	& \ctrl{3}				& \qw		& \qw	&&& \qw		& \targ{}				& \ctrl{3}	& \targ{}	& \targ{}\gategroup[4,steps=7,style={rounded corners,draw=none,fill=blue!15,inner ysep=2pt,inner xsep=24pt,xshift=24pt,yshift=3pt},background,label style={label position=below,xshift=8pt,yshift=10pt}]{{$\mathit{CCR}_y(2\theta)$}}
																																	& \RyT{-}{2}	& \targ{}	& \RyT{}{2} 	& \targ{}	& \RyT{-}{2}	& \targ{}	& \RyT{}{2}	& \ctrl{3}	& \ctrl{-1}				& \qw		& \ctrl{-1}	& \qw		\\[-10pt]
		&					&		&			&					&		&	&=&&&&&&&&&&&&&&																																			\\[-10pt]	
		& \ctrl{1}				& \qw		& \control{}		& \qw					& \qw		& \qw	&&& \ctrl{1}	& \ctrl{-2}				& \qw		& \ctrl{-2}	& \ctrl{-2}	& \qw		& \qw		& \qw		& \ctrl{-2}	& \qw		& \qw		& \qw		& \qw		& \qw					& \qw		& \qw		& \qw		\\
		& \targ{}				& \targ{}	& \ctrl{-3}		& \targ{}				& \targ{}	& \qw	&&& \targ{}	& \qw					& \targ{}	& \qw		& \qw		& \qw		& \ctrl{-3}	& \qw		& \qw		& \qw		& \ctrl{-3}	& \qw		& \targ{}	& \qw					& \targ{}	& \qw		& \qw
		\end{tikzcd}
	};
	\end{tikzpicture}\\[-2ex]
	\caption{Implementing a $\CNOT$-conjugated two-control $\mathit{CCR_y}(2\theta)$ rotation gate (as given in a $\blockii{\ell}$ gate, see Figure~\ref{fig:U53}) with four single-qubit $R_y(\pm\tfrac{\theta}{2})$ gates and four $\CNOT$s, one of which can be cancelled by rearranging the last $\CNOT$ of the preceding $\blocki$/$\blockii{\ell-1}$ gate.}
	\label{fig:blockii}
\end{figure}
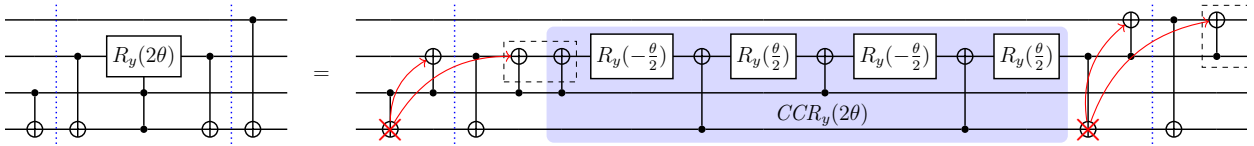

\begin{proof}[Proof (Arbitrary 2-qubit gates).]
	Note that an alternate way to prepare a Dicke state $\dicke{n}{k}$ is to prepare the Dicke state $\dicke{n}{n-k}$ followed by $X$-gates on each qubit, as $\dicke{n}{k} = X^{\otimes n}\dicke{n}{n-k}$.
	Thus we prove size $\bigO(kn)$ for Dicke states $\dicke{n}{k}$, implying size $\bigO((n-k)n)$ for $\dicke{n}{k}$, too.

	We first show that the depth of our circuit construction is linear: 
	The structure of each $\SCS_{n,k}$ implementation is a stair of 2-qubit blocks interacting with its bottom qubit $n$. These stairs can be ``pushed into each other''.
	In particular, the 3-qubit gate $\blockii{k}$ of $\SCS_{n,k}$ acts on qubits $n-k, n-k+1$ and $n$. It can therefore be run in parallel with $k^* := \lfloor \tfrac{k+1}{3}\rfloor-1$ many other 3-qubit gates,
	namely gate $\blockii{k-3}$ of $\SCS_{n-1,k}$ (acting on qubits $n-k+2,n-k+3,n-1$) as well as gates $\blockii{k-6}, \ldots, \blockii{k-3 k^*}$ of $\SCS_{n-2,k}, \ldots, \SCS_{n-k^*,k}$, respectively.
	Since we can parallelize $k^* \in \bigO(k)$ stairs, the total depth is linear in the depth of gates $\blocki,\blockii{\ell}$ (constant) and the number of gates $\blocki,\blockii{\ell}$ ($\bigO(kn)$) divided by $k^*$,
	yielding an overall depth of $\bigO(n)$. 

	In light of a possible implementation, we prove the circuit size in Theorem~\ref{thm:DS-circuit} by compiling $U_{n,k}$ down to at most $5kn + \bigO(n)$ $\CNOT$-gates and $4kn + \bigO(n)$ arbitrary precision $R_y$-gates.
	To build $U_{n,k}$ from $\SCS$ unitaries, we need a total of $n-1$ many 2-qubit gates $\blocki$ and
	$(n-k)\cdot (k-1) + \sum_{i=3}^k (i-2) = kn - \tfrac{k^2}{2} +\bigO(n)$ many 3-qubit gates $\blockii{\ell}$, see Figure~\ref{fig:U53}.
	It remains to show that $\blockii{\ell}$-gates can be implemented with 5 $\CNOT$ gates and 4 $R_y$ gates.
	We provide such an implementation in Figure~\ref{fig:blockii}: A two-controlled $\mathit{CCR_y}(2\theta)$ rotation gate is easily seen to be implemented with 4 $R_y(\pm\tfrac{\theta}{2})$ rotation gates and 
	4 $\CNOT$s, the first one of which we can cancel by rearranging the preceeding conjugating $\CNOT$ gate.
\end{proof}

\begin{proof}[Proof (LNN architectures).] 
	In order to have the same asymptotic behaviour in terms of size and depth on LNN architectures, we need to slightly adapt our circuit.
	First note any 3-qubit gate that spans neighboring triples of qubits can be implemented with a constant number of one- and two-qubit LNN gates.
	Hence, when building unitaries $\SCS_{n,k}$ from $\blockii{\ell}$-gates, we sift up qubit $n$ by always swapping it one position upwards such that it lies on wire $n-\ell+2$.

	In order to achieve size $\bigO(kn)$ and depth $\bigO(n)$ we group $\SCS$ unitaries into $\lfloor \tfrac{n}{k} \rfloor -1$ groups of $k$ consecutive unitaries, 
	see e.g. Figure~\ref{fig:lnn} where we group $\SCS_{15,5}, \SCS_{14,5}, \ldots, \SCS_{11,5}$.
	As mentioned, to implement $\SCS_{n,k}$ on wires $n,\ldots,n-k$ we sift up qubit $n$ until it reaches wire $n-k+2$. We then \emph{continue} to let it sift up to wire $n-2k+1$. 
	In parallel we implement $\SCS_{n-1,k}$ (due to the sifting up of qubit $n$ also on wires $n,\ldots, n-k$) in a similar fashion. Again we let qubit $n-1$ sift up to wire $n-2k+1$. 
	Implementing this in parallel for all $\SCS_{n,k}, \ldots, \SCS_{n-k+1,k}$ needs size $\bigO(k^2)$ and depth $\bigO(k)$ and ends
	with qubits $n-k+1, \ldots, n$ on wires $n-2k+1,\ldots,n-k$ and vice versa, see the the red-dashed slice in Figure~\ref{fig:lnn}.
	Using another $\bigO(k^2)$ gates and $\bigO(k)$ depth we let qubits $n-k+1,\ldots,n$ sift down in parallel back to their original wires, which also moves qubits $n-2k+1,\ldots,n-k$ back to their positions. 
	The same technique can be used for all blocks of $\SCS$ gates on their respective positions.

	We are left with a group of $(n \bmod k) + k < 2k$ $\SCS$ unitaries. For those, we always sift up the respective qubits up to wire $1$. Again, we need $\bigO(k^2)$ gates and a depth of order $\bigO(k)$. 
	Overall, we need $\bigO(\tfrac{n}{k}\cdot k^2) = \bigO(kn)$ gates and $\bigO(\tfrac{n}{k}\cdot k) = \bigO(n)$ depth. 
	Note that the sift-down operations (which we did not apply on the last group of qubits) are not necessary for correctness -- they are needed to prevent a blow-up from $\bigO(kn)$ to $\bigO(n^2)$ gates overall.
\end{proof}

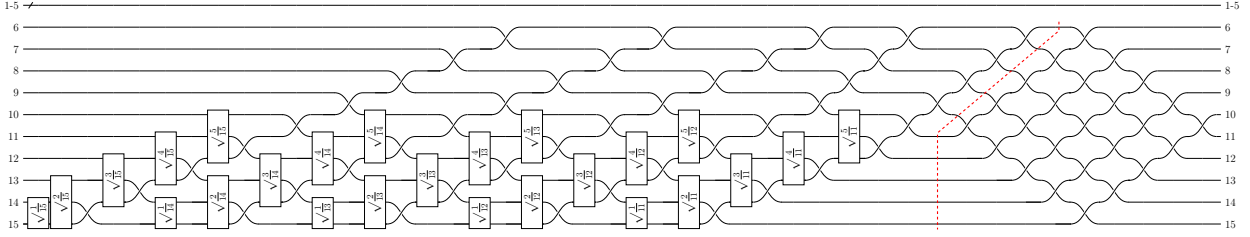
\begin{figure}[t!]
	\centering
	\newcommand{\mg}[3]{\gate[wires=#1]{\rotatebox{90}{\large$\mathclap{\surd\tfrac{#2}{#3}}$}}}
	\newcommand{\swapg}{\gate[swap,style={draw=none}]{}}
	\newcommand{\nog}{\gate[style={draw=none}]{\hspace*{60pt}}}
	\begin{tikzpicture}
	\clip (-8.2,1.5) rectangle (8.3,-1.6);
	\node[scale=0.345]{
		\begin{tikzcd}[row sep={24pt,between origins}, column sep=8pt, execute at end picture={
				\draw[very thick,red,dashed] ($(\tikzcdmatrixname-2-40.center)+(4pt,6pt)$) -- +(0pt,-10pt) -- ($(\tikzcdmatrixname-7-36.center)+(0pt,4pt)$) -- ($(\tikzcdmatrixname-11-36)+(0pt,-6pt)$);	
			}]
			\lstick{1-5}	& \qw\qwbundle	& \qw		& \qw		& \qw		& \qw		& \qw		& \qw		& \qw		& \qw		& \qw		& \qw		& \qw		& \qw		& \qw		& \qw		& \qw		& \qw		& \qw		& \qw		& \qw		& \qw		& \qw		& \qw		& \qw		& \qw		& \qw		& \qw		& \qw		& \qw		& \qw		& \qw		& \qw		& \qw		& \qw		& \qw		& \qw		& \qw		& \qw		& \qw		& \qw		& \qw		& \qw		& \qw		& \qw		& \qw\rstick{1-5}	\\
			\lstick{6}	& \qw		& \qw		& \qw		& \qw		& \qw		& \qw		& \qw		& \qw		& \qw		& \qw		& \qw		& \qw		& \qw		& \qw		& \qw		& \qw		& \qw		& \qw		& \swapg	& \qw		& \qw		& \qw		& \qw		& \qw		& \swapg	& \qw		& \qw		& \qw		& \qw		& \qw		& \swapg	& \qw		& \qw		& \swapg	& \qw		& \qw		& \qw		& \swapg	& \qw		& \swapg	& \qw		& \qw		& \qw		& \qw		& \qw\rstick{6}		\\
			\lstick{7}	& \qw		& \qw		& \qw		& \qw		& \qw		& \qw		& \qw		& \qw		& \qw		& \qw		& \qw		& \qw		& \qw		& \qw		& \qw		& \qw		& \swapg	& \qw		& \qw		& \qw		& \qw		& \qw		& \swapg	& \qw		& \qw		& \qw		& \qw		& \qw		& \swapg	& \qw		& \qw		& \qw		& \swapg	& \qw		& \qw		& \qw		& \swapg	& \qw		& \swapg	& \qw		& \swapg	& \qw		& \qw		& \qw		& \qw\rstick{7}		\\
			\lstick{8}	& \qw		& \qw		& \qw		& \qw		& \qw		& \qw		& \qw		& \qw		& \qw		& \qw		& \qw		& \qw		& \qw		& \qw		& \swapg	& \qw		& \qw		& \qw		& \qw		& \qw		& \swapg	& \qw		& \qw		& \qw		& \qw		& \qw		& \swapg	& \qw		& \qw		& \qw		& \qw		& \swapg	& \qw		& \qw		& \qw		& \swapg	& \qw		& \swapg	& \qw		& \swapg	& \qw		& \swapg	& \qw		& \qw		& \qw\rstick{8}		\\
			\lstick{9}	& \qw		& \qw		& \qw		& \qw		& \qw		& \qw		& \qw		& \qw		& \qw		& \qw		& \qw		& \qw		& \swapg	& \qw		& \qw		& \qw		& \qw		& \qw		& \swapg	& \qw		& \qw		& \qw		& \qw		& \qw		& \swapg	& \qw		& \qw		& \qw		& \qw		& \qw		& \swapg	& \qw		& \qw		& \qw		& \swapg	& \qw		& \swapg	& \qw		& \swapg	& \qw		& \swapg	& \qw		& \swapg	& \qw		& \qw\rstick{9}		\\
			\lstick{10}	& \qw		& \qw		& \qw		& \qw		& \qw		& \qw		& \qw		& \mg{3}{5}{15}	& \qw		& \qw		& \swapg	& \qw		& \qw		& \mg{3}{5}{14}	& \qw		& \qw		& \swapg	& \qw		& \qw		& \mg{3}{5}{13}	& \qw		& \qw		& \swapg	& \qw		& \qw		& \mg{3}{5}{12}	& \qw		& \qw		& \swapg	& \qw		& \qw		& \mg{3}{5}{11}	& \qw		& \swapg	& \qw		& \swapg	& \qw		& \swapg	& \qw		& \swapg	& \qw		& \swapg	& \qw		& \swapg	& \qw\rstick{10}	\\
			\lstick{11}	& \qw		& \qw		& \qw		& \qw		& \qw		& \mg{3}{4}{15}	& \qw		& \qw		& \swapg	& \qw		& \qw		& \mg{3}{4}{14}	& \qw		& \qw		& \swapg	& \qw		& \qw		& \mg{3}{4}{13}	& \qw		& \qw		& \swapg	& \qw		& \qw		& \mg{3}{4}{12}	& \qw		& \qw		& \swapg	& \qw		& \qw		& \mg{3}{4}{11}	& \qw		& \qw		& \swapg	& \qw		& \qw		& \qw		& \swapg	& \qw		& \swapg	& \qw		& \swapg	& \qw		& \swapg	& \qw		& \qw\rstick{11}	\\
			\lstick{12}	& \qw		& \qw		& \qw		& \mg{3}{3}{15}	& \qw		& \qw		& \swapg	& \qw		& \qw		& \mg{3}{3}{14}	& \qw		& \qw		& \swapg	& \qw		& \qw		& \mg{3}{3}{13}	& \qw		& \qw		& \swapg	& \qw		& \qw		& \mg{3}{3}{12}	& \qw		& \qw		& \swapg	& \qw		& \qw		& \mg{3}{3}{11} & \qw		& \qw		& \swapg	& \qw		& \qw		& \qw		& \qw		& \qw		& \qw		& \swapg	& \qw		& \swapg	& \qw		& \swapg	& \qw		& \qw		& \qw\rstick{12}	\\
			\lstick{13}	& \qw		& \mg{3}{2}{15}	& \qw		& \qw		& \swapg	& \qw		& \qw		& \mg{3}{2}{14}	& \qw		& \qw		& \swapg	& \qw		& \qw		& \mg{3}{2}{13}	& \qw		& \qw		& \swapg	& \qw		& \qw		& \mg{3}{2}{12}	& \qw		& \qw		& \swapg	& \qw		& \qw		& \mg{3}{2}{11}	& \qw		& \qw		& \swapg	& \qw		& \qw		& \qw		& \qw		& \qw		& \qw		& \qw		& \qw		& \qw		& \swapg	& \qw		& \swapg	& \qw		& \qw		& \qw		& \qw\rstick{13}	\\		
			\lstick{14}	& \mg{2}{1}{15}	& \qw		& \swapg	& \qw		& \qw		& \mg{2}{1}{14}& \qw		& \qw		& \swapg	& \qw		& \qw		& \mg{2}{1}{13}	& \qw		& \qw		& \swapg	& \qw		& \qw		& \mg{2}{1}{12}	& \qw		& \qw		& \swapg	& \qw		& \qw		& \mg{2}{1}{11}	& \qw		& \qw		& \swapg	& \qw		& \qw		& \qw		& \qw		& \qw		& \qw		& \qw		& \qw		& \qw		& \qw		& \qw		& \qw		& \swapg	& \qw		& \qw		& \qw		& \qw		& \qw\rstick{14}	\\
			\lstick{15}	& \qw		& \qw		& \qw		& \qw		& \qw		& \qw		& \qw		& \qw		& \qw		& \qw		& \qw		& \qw		& \qw		& \qw		& \qw		& \qw		& \qw		& \qw		& \qw		& \qw		& \qw		& \qw		& \qw		& \qw		& \qw		& \qw		& \qw		& \qw		& \qw		& \qw		& \qw		& \qw		& \qw		& \qw		& \qw		& \qw		& \qw		& \qw		& \qw		& \qw		& \qw		& \qw		& \qw		& \qw		& \qw\rstick{15}	   
		\end{tikzcd}
	};
	\end{tikzpicture}\\[-2ex]
	\caption{Linear circuit depth on LLN architectures is achieved by $\lfloor\tfrac{n}{k}\rfloor$ blocks of $\bigO(k)$ $\SCS$ unitaries which can be implemented using $\bigO(k^2)$ gates and $\bigO(k)$ depth 
	(here: $n=15, k=5$).}
	\label{fig:lnn}
\end{figure}

\section{Symmetric Pure States and Quantum Compression}
\label{sec:superposition-compression}

Our inductive approach yields (for $k=n$) a unitary $U_{n,n}$ which -- with $\bigO(n^2)$ gates and $\bigO(n)$ depth -- 
can be used to prepare any Dicke state $\dicke{n}{\ell},\ 1\leq \ell \leq n$ for the respective input $\ket{0}^{\otimes n-\ell}\ket{1}^{\otimes \ell}$.
Therefore, every superposition of these input states leads to a superposition of Dicke states.
In the following, we show how this can be used to 
(i) prepare arbitrary symmetric pure $n$-qubit states in linear depth $\bigO(n)$, and to
(ii) compress symmetric pure $n$-qubit states into $\lceil \log(n+1)\rceil$ qubits in quasilinear depth $\tilde{\bigO}(n)$ using the reverse unitary~$U_{n,n}^{\dagger}$.

\subsection{Symmetric Pure States}

As the $n+1$ different Dicke states $\dicke{n}{\ell}$ form an orthonormal basis of the fully symmetric subspace of all pure $n$-qubit states,
every symmetric pure state can be expanded in terms of Dicke states~\cite{Bastin2009}, i.e. in the form
\[ \sum_{\ell} e^{i \phi_\ell} \alpha_{\ell} \dicke{n}{\ell} \]
with magnitudes $\alpha_{\ell} \in [0,1]$, $\alpha_0^2 + \ldots + \alpha_n^2 = 1$ and phases $\phi_{\ell} \in [0, 2\pi)$, $\phi_0 = 0$.
We show:

\begin{theorem}
	Every symmetric pure $n$-qubit state can be prepared with a circuit of size $\bigO(n^2)$ and depth $\bigO(n)$, 
	even on Linear Nearest Neighbor architectures.
	\label{thm:sym-states}
\end{theorem}

\begin{proof}
	By Theorem~\ref{thm:DS-circuit}, given as input the state $\sum_{\ell} e^{i \phi_\ell} \alpha_{\ell} \ket{0}^{\otimes n-\ell}\ket{1}^{\otimes \ell}$,
	the unitary $U_{n,n}$ prepares $\sum_{\ell} e^{i \phi_\ell} \alpha_{\ell} \dicke{n}{\ell}$ on LNN architectures using $\bigO(n^2)$ gates and $\bigO(n)$ depth.
	We prove that this input state can be constructed in linear depth and size.
	
	To this end, we define magnitudes $\beta_{\ell}:= \alpha_{\ell} (1-\alpha_0^2-\ldots-\alpha_{\ell-1})^{-1/2}$ and angles $\psi_0 := 0$, $\psi_{\ell} := \phi_{\ell} - \phi_{\ell-1}$.
	The original values $\alpha, \phi$ relate to the parameters $\beta, \psi$ as  
	$\phi_{\ell} = \sum_{k=0}^{\ell} \psi_{k}$ and $\alpha_{\ell} = \beta_{\ell} \cdot \prod_{k=0}^{\ell-1} \sqrt{1-\beta_k^2}$.
	As already introduced in Section~\ref{sec:circuit-size}, we use $Y$-rotation gates $R_y(2\cos^{-1} \beta)$ to map $\ket{0} \rightarrow \beta \ket{0} + \sqrt{1-\beta^2}\ket{1}$.
	Additionally, we use phase shift gates $R_{\psi} = \left(\begin{smallmatrix} 1&0\\ 0& e^{i\psi} \end{smallmatrix}\right)$ to map $\ket{1} \rightarrow  e^{i \psi} \ket{1}$.

	We start with a rotation $R_y(2\cos^{-1} \beta_0)$ on the $n$-th qubit. This is followed by a linear-depth stair of controlled $R_y(2\cos^{-1} \beta_{\ell})$-rotations
	on the $(n-\ell)$th qubit, controlled by the previous qubit $n-\ell+1$ being in state $\ket{1}$, as shown in Figure~\ref{fig:superposition-compression}~(left).	
	Finally, we add the correct phases using a layer of $R_{\psi_{\ell}}$-phase shifts on respective qubits $n-\ell+1$, yielding the desired state:
	\begin{align*}
		\ket{0}^{\otimes n}	
		& \xrightarrow[\hspace*{8ex}]{\beta_0}				\alpha_0 \ket{0}^{\otimes n} + \sqrt{1-\alpha_0^2} \ket{0}^{\otimes n-1}\ket{1}									\\
		& \xrightarrow[\hspace*{8ex}]{\beta_1}				\alpha_0 \ket{0}^{\otimes n} + \alpha_1 \ket{0}^{\otimes n-1}\ket{1} + \sqrt{1-\alpha_0^2-\alpha_1^2} \ket{0}^{\otimes n-2}\ket{11}		\\
		& \xrightarrow[\hspace*{8ex}]{\beta_2, \ldots, \beta_{n-1}}	\sum\nolimits_{\ell} \alpha_{\ell} \ket{0}^{\otimes n-\ell}\ket{1}^{\otimes \ell}								
		\xrightarrow[\hspace*{8ex}]{\psi_1, \ldots, \psi_n}		\sum\nolimits_{\ell} e^{i\phi_{\ell}} \alpha_{\ell} \ket{0}^{\otimes n-\ell}\ket{1}^{\otimes \ell}. 	\qedhere
	\end{align*}
\end{proof}

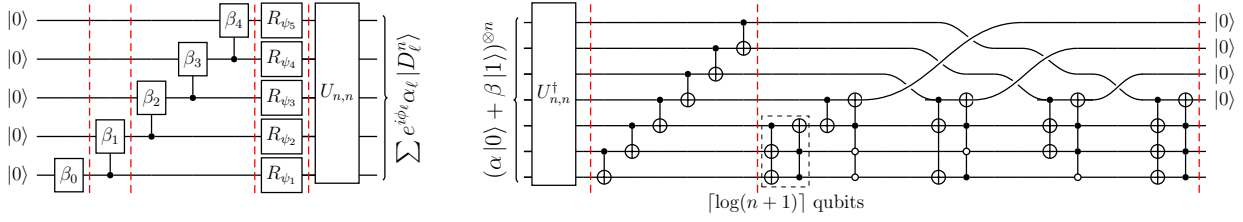
\begin{figure}[t!]
	\centering
	\newcommand{\nogate}{\gate[style={draw=none}]{\phantom{U}}}
	\begin{tikzpicture}
	\clip (-8.2,1.4) rectangle (8.3,-1.5);
	\node[scale=0.61]{
		\begin{tikzcd}[row sep={24pt,between origins}]
			\lstick{\ket{0}}	& \qw			& \qw			& \qw			& \qw			& \gate{\beta_4}	& \gate{R_{\psi_5}}\slice{}	& \gate[wires=5]{U_{n,n}}	& \qw \rstick[wires=5]{\rotatebox{90}{\large $\sum e^{i \phi_{\ell}} \alpha_{\ell} \dicke{n}{\ell}$}} \\		
			\lstick{\ket{0}}	& \qw			& \qw			& \qw			& \gate{\beta_3}	& \ctrl{-1}		& \gate{R_{\psi_4}}		& \qw				& \qw \\		
			\lstick{\ket{0}}	& \qw			& \qw\slice{}		& \gate{\beta_2}	& \ctrl{-1}		& \qw\slice{}		& \gate{R_{\psi_3}}		& \qw				& \qw \\		
			\lstick{\ket{0}}	& \qw\slice{}		& \gate{\beta_1}	& \ctrl{-1}		& \qw			& \qw			& \gate{R_{\psi_2}}		& \qw				& \qw \\		
			\lstick{\ket{0}}	& \gate{\beta_0}	& \ctrl{-1}		& \qw			& \qw			& \qw			& \gate{R_{\psi_1}}		& \qw				& \qw  
		\end{tikzcd}
		\qquad
		\begin{tikzcd}[row sep={16pt,between origins}, column sep=8pt, execute at end picture={
					\node[fit=(\tikzcdmatrixname-5-10)(\tikzcdmatrixname-7-11),draw,dashed,inner sep=4pt,label={[yshift=-0cm]below:$\lceil\log(n+1)\rceil$ qubits}] {};
				\draw[thick] (\tikzcdmatrixname-3-14) edge[in=180, out=0] (\tikzcdmatrixname-4-16);
				\draw[thick] (\tikzcdmatrixname-2-15) edge[in=180, out=0] (\tikzcdmatrixname-3-17);
				\draw[thick] (\tikzcdmatrixname-1-16) edge[in=180, out=0] (\tikzcdmatrixname-2-18);
				\draw[line width=4pt,draw=white] (\tikzcdmatrixname-4-13) edge[in=180, out=0] (\tikzcdmatrixname-1-19);
				\draw[thick] (\tikzcdmatrixname-4-13) edge[in=180, out=0] (\tikzcdmatrixname-1-19);
				\draw[thick] (\tikzcdmatrixname-3-18) edge[in=180, out=0] (\tikzcdmatrixname-4-20);
				\draw[thick] (\tikzcdmatrixname-2-19) edge[in=180, out=0] (\tikzcdmatrixname-3-21);
				\draw[line width=4pt,draw=white] (\tikzcdmatrixname-4-17) edge[in=180, out=0] (\tikzcdmatrixname-2-21);
				\draw[thick] (\tikzcdmatrixname-4-17) edge[in=180, out=0] (\tikzcdmatrixname-2-21);
			}]
			\lstick[wires=7]{\rotatebox{90}{\large $(\alpha\ket{0}+\beta\ket{1})^{\otimes n}$}}
			& \gate[wires=7]{U_{n,n}^{\dagger}}	& \qw\slice{}	& \qw		& \qw		& \qw		& \qw		& \qw		& \ctrl{1}	& \qw		& \qw		& \qw		& \qw		& \qw		& \qw		& \qw 		& 		&		& 		& \qw		& \qw		& \qw					& \qw		& \qw\slice{}	& \qw\rstick{\ket{0}}	\\
			& \qw					& \qw		& \qw		& \qw		& \qw		& \qw		& \ctrl{1}	& \targ{}	& \qw		& \qw		& \qw		& \qw		& \qw		& \qw		& 		& 		&		& \qw		& 		& 		& \qw					& \qw		& \qw		& \qw\rstick{\ket{0}}	\\
			& \qw					& \qw		& \qw		& \qw		& \qw		& \ctrl{1}	& \targ{}	& \qw\slice{}	& \qw		& \qw		& \qw		& \qw		& \qw		& 		& 		& 		& \qw		& 		& 		& 		& \gate[swap,style={draw=none}]{}	& \qw		& \qw		& \qw\rstick{\ket{0}}	\\
			& \qw					& \qw		& \qw		& \qw		& \ctrl{1}	& \targ{}	& \qw		& \qw		& \qw		& \qw		& \ctrl{1}	& \targ{}	& \nogate	& \nogate	& \ctrl{3}	& \targ{}	& \nogate	& \nogate	& \ctrl{2}	& \targ{}	& \qw					& \ctrl{3}	& \targ{}	& \qw\rstick{\ket{0}}	\\
			& \qw					& \qw		& \qw		& \ctrl{1}	& \targ{}	& \qw		& \qw		& \qw		& \ctrl{2}	& \targ{}	& \targ{}	& \ctrl{-1}	& \qw		& \qw 		& \targ{}	& \ctrl{-1}	& \qw		& \qw		& \targ{}	& \ctrl{-1}	& \qw					& \targ{}	& \ctrl{-1}	& \qw			\\
			& \qw					& \qw		& \ctrl{1}	& \targ{}	& \qw		& \qw		& \qw		& \qw		& \targ{}	& \ctrl{-1}	& \qw		& \octrl{-1}	& \qw		& \qw		& \qw		& \octrl{-1}	& \qw		& \qw		& \targ{}	& \ctrl{-1}	& \qw					& \targ{}	& \ctrl{-1}	& \qw			\\
			& \qw					& \qw		& \targ{}	& \qw		& \qw		& \qw		& \qw		& \qw		& \targ{}	& \ctrl{-1}	& \qw		& \octrl{-1}	& \qw		& \qw		& \targ{}	& \ctrl{-1}	& \qw		& \qw		& \qw		& \octrl{-1}	& \qw					& \targ{}	& \ctrl{-1}	& \qw		 	 
		\end{tikzcd}
	};		
	\end{tikzpicture}\\[-2ex]
	\caption{(left) Arbitrary superposition of Dicke states, using $Y$-rotations $R_y(2\cos^{-1} \beta_{\ell})$ with $\beta_{\ell} = \surd \tfrac{\alpha_{\ell}^2}{1-\alpha_0^2-\ldots-\alpha_{\ell-1}^2}$ and phase-shift gates $R_{\psi_\ell}$ with $\psi_{\ell} = \phi_{\ell}-\phi_{\ell-1}$, followed by unitary $U_{n,n}$. 
	(right) Efficient compression of $n$ identical qubits into $\lceil \log(n+1)\rceil$ qubits using unitary $U_{n,n}^{\dagger}$.}
	\label{fig:superposition-compression}
\end{figure}

\subsection{Quantum Compression}

As symmetric pure states live in the $(n+1)$-dimensional symmetric subspace of the full Hilbert space, they can be described with exponentially 
fewer dimensions than general multi-qubit states. This is the idea behind the quantum Schur-Weyl transform~\cite{Bacon2006}, which separates 
the permutation information from the angular momentum information of a state. Applied to a symmetric pure state, it will compress the angular momentum information
into only $\lceil \log(n+1) \rceil$ qubits, while the rest of the qubits (the trivial permutation information) can be discarded without loss of information.    

A previous approach to implement this transform for symmetric states~\cite{Plesch2010} gave a high-level description of a circuit of size and depth $\Theta(n^2)$ that needs no ancillas. 
The major part is a transformation of Dicke states $\dicke{n}{\ell}$ to a one-hot encoding $\ket{0}^{\otimes \ell-1}\ket{1}\ket{0}^{\otimes n-\ell}$ of their Hamming weight $\ell$.
Substituting this part with our unitary in reverse, $U_{n,n}^{\dagger}$, improves the overall circuit depth to quasilinear:\footnote{The referenced paper~\cite{Plesch2010}
	provides no compilation down to standard gates and no analysis of the depth of the circuit. The latter is found together with a step-by-step comparison of our
approach in Appendix~\ref{app:comparison}.}

\begin{theorem}
	\label{thm:compression}
	Every symmetric pure $n$-qubit state can be compressed into $\lceil \log(n+1)\rceil$ qubits with a circuit of size $\bigO(n^2)$ and depth $\tilde{\bigO}(n)$,
	even on Linear Nearest Neighbor architectures.
\end{theorem}

We illustrate our approach with a particular interesting symmetric pure state, the separable state 
$(\alpha\ket{0} + \beta\ket{1})^{\otimes n} = \sum \alpha^{n-\ell}\beta^{\ell} \tbinom{n}{\ell}^{1/2} \dicke{n}{\ell}$,
whose compression has been implemented experimentally for $n=3$~\cite{Rozema2014}.
An implementation of our approach for $n=5$ qubits in Quirk~\cite{quirk} can be found 
\href{https://algassert.com/quirk#circuit=

\begin{proof}
	Our compression circuit starts with the reverse unitary $U_{n,n}^{\dagger}$ (using size $\bigO(n^2)$ and depth $\bigO(n)$). 
	It is followed by a mapping of states $\ket{0}^{\otimes n-\ell}\ket{1}^{\otimes \ell}$ to the one-hot encoding $\ket{0}^{\otimes n-\ell}\ket{1}\ket{0}^{\otimes \ell-1}$,
	which can be implemented with size and depth $\bigO(n)$ with a simple stair of $\CNOT$-gates with control $n-\ell$ and target $n-\ell+1$ for increasing $\ell$, 
	see Figure~\ref{fig:superposition-compression}~(right). Finally, the one-hot encoding $\ket{0}^{\otimes n-\ell}\ket{1}\ket{0}^{\otimes \ell-1}$ is mapped to 
	the binary encoding $\ket{\ell}$ of $\ell$ (with padded leading zeroes), as illustrated in Figure~\ref{fig:superposition-compression}~(right):
	\begin{align*}
		\sum \alpha^{n-\ell}\beta^{\ell} \tbinom{n}{\ell}^{1/2} \dicke{n}{\ell}	
		& \xrightarrow[\hspace*{8ex}]{U_{n,n}^{\dagger}}				\sum \alpha^{n-\ell}\beta^{\ell} \tbinom{n}{\ell}^{1/2} \ket{0}^{\otimes n-\ell}\ket{1}^{\otimes \ell}			\\
		& \xrightarrow[\hspace*{8ex}]{\substack{\CNOT\\\text{stair}}}			\alpha^n\ket{0}^{\otimes n} + \sum_{\ell>0} \alpha^{n-\ell}\beta^{\ell} \tbinom{n}{\ell}^{1/2} \ket{0}^{\otimes n-\ell}\ket{1}\ket{0}^{\otimes \ell-1}		\\
		& \xrightarrow[\hspace*{8ex}]{\substack{\text{encoding}\\\text{change}}}	\sum \alpha^{n-\ell}\beta^{\ell} \tbinom{n}{\ell}^{1/2} \ket{\ell}.
	\end{align*}
	It remains to show that mapping each state $\ket{0}^{\otimes n-\ell}\ket{1}\ket{0}^{\otimes \ell-1}$ into the bottom $\lceil \log(n+1)\rceil$ qubits encoding $\ket{\ell}$
	can be implemented with a circuit of size $\bigO(n^2)$ and depth $\tilde{\bigO}(n)$.
	
	This is done in the following way, for increasing $\ell$: 
	First, controlled on qubit $n-\ell$ we $\CNOT$ into the up to $\lceil \log(n+1) \rceil$ target bottom qubits which represent the number $\ell$ in binary 
	(for $n$ and $n-1$ these are the qubits themselves, on which we perform no operation). 
	Then, controlled on the binary representation in the last $\ell$ qubits (not including padded $0$s), we perform a single multi-control Toffoli on the $(n-\ell)$th qubit as target.
	An implementation of $m$-control Toffoli gates with $\CNOT$ and single-qubit gates requires at least $2m$ $\CNOT$ gates~\cite{Shende2009}, 
	and it is known that a $O(m)$ $\CNOT$ and single-qubit gates are sufficient~\cite{Barenco1995}, even if no ancilla qubits are present~\cite{Gidney2015}.
	This immediately gives us $\bigO(n\log n)$ gates and $\bigO(n\log n)$ depth.

	For LNN architectures, using SWAP gates we let each processed qubit $n-\ell$ (for $\ell \geq \lceil \log(n+1)\rceil$) sift up to the top wires, 
	in order to bring the next qubit $n-\ell-1$ into direct neighborhood of the bottom $\lceil \log(n+1)\rceil$ qubits. Since sifting up operations can be done in parallel,
	this requires $\bigO(n^2)$ gates but only $\bigO(n)$ depth, see Figure~\ref{fig:superposition-compression}~(right).
	As mentioned, the $\bigO(n)$ many multi-target and multi-control operations can be implemented using $\bigO(\log n)$ (arbitrary 2-qubit) $\CNOT$ gates each.
	However, these might not be between neighboring qubits. Moving qubits to neighboring positions and back needs up to $\bigO(\log n)$ SWAP operations. 
	In total, this adds $\bigO(n\log^2 n)$ gates and depth, concluding our proof.	
\end{proof}

\section{Conclusions}
We presented a deterministic quantum circuit for the preparation of
Dicke states  $\dicke{n}{k}$ with depth $\bigO(n)$ and  $\bigO(kn)$
gates in total. We showed that these bounds hold for Linear Nearest
Neighbor architectures, that the circuit can be extended to prepare
arbitrary symmetric pure states, and that we can use it for quantum
compression. For future work, the main open problem is that of
characterizing the set of quantum states that can be prepared in
polynomial time, of which Dicke states are one example.

\paragraph{Acknowledgments.} We would like to thank Yiğit Subaşı for helpful discussions.
Research presented in this article was supported by the Center for Nonlinear Studies CNLS and the Laboratory Directed Research and Development program of Los Alamos National Laboratory under project number 20190495DR.
\hfill LA-UR-19-22718

\nocite{quantikz}
\bibliographystyle{plainurl}
\bibliography{DS-bib.bib}

\newpage
\appendix
\section{Quantum Compression: Comparison with other Work}
\label{app:comparison}

\begin{figure}[ht!]
	\centering
	\newcommand{\Ugate}[2]{\gate[wires=#1,style={yshift=-2}]{\mathit{#2}}}
	\newcommand{\Upgate}[2]{\gate[wires=#1,style={dashed,yshift=-2}]{\mathit{#2}}}
	\begin{tikzpicture}
	\clip (-8.2,1.4) rectangle (8.3,-1.5);
	\node[scale=0.74]{
		\begin{tikzcd}[row sep={20pt,between origins}, execute at end picture={
					\draw[green!50!black,->] ([yshift=4pt]\tikzcdmatrixname-1-8.north) to[bend right=70] node[below] {$\checkmark$} ([yshift=6pt]\tikzcdmatrixname-1-7.north);
			}]
			\lstick[wires=5]{\rotatebox{90}{\large\dicke{5}{3}}}	
			& \Ugate{2}{i}	& \Ugate{3}{ii}	& \qw			& \Ugate{2}{ii}			& \qw		& \qw				& \Ugate{2}{ii}			& \qw				& \qw		& \qw			& \qw \rstick{\ket{0}}		\\
			& \qw		& \qw		& \Ugate{2}{i}		& \vqw{2}			& \Ugate{3}{ii}	& \qw				& \vqw{3}			& \Ugate{2}{ii}			& \qw		& \qw			& \qw \rstick{\ket{0}}		\\
			& \qw		& \qw		& \qw			& \qw				& \qw		& \Ugate{2}{i}			& \qw				& \vqw{2}			& \Ugate{3}{ii}	& \qw			& \qw \rstick{\ket{1}}		\\
			& \qw		& \qw		& \qw			& \gate[style={yshift=2pt}]{}	& \qw		& \qw				& \qw				& \qw				& \qw		& \Ugate{2}{i}		& \qw \rstick{\ket{1}}		\\
			& \qw		& \qw		& \qw			& \qw				& \qw		& \qw				& \gate[style={yshift=2pt}]{}	& \gate[style={yshift=2pt}]{}	& \qw		& \qw			& \qw \rstick{\ket{1}}		
		\end{tikzcd}
		\qquad
		\begin{tikzcd}[row sep={20pt,between origins}, execute at end picture={
				\draw[red,->] ([yshift=4pt]\tikzcdmatrixname-1-8.north) to[bend right=70] node {$\times$} ([yshift=4pt]\tikzcdmatrixname-1-7.north);
			}]
			\lstick[wires=5]{\rotatebox{90}{\large\dicke{5}{3}}}	
			& \Ugate{2}{I}	& \Ugate{3}{II}	& \Upgate{3}{III}	& \Ugate{2}{II}			& \qw		& \Upgate{1}{\phantom{I}}	& \Ugate{2}{II}			& \qw				& \qw		& \Upgate{1}{}		& \qw \rstick{\ket{0}}		\\
			& \qw		& \qw		& \qw			& \vqw{2}			& \Ugate{3}{II}	& \qw				& \vqw{3}			& \Ugate{2}{II}			& \qw		& \qw			& \qw \rstick{\ket{0}}		\\
			& \qw		& \qw		& \qw			& \qw				& \qw		& \Upgate{2}{III}		& \qw				& \vqw{2}			& \Ugate{3}{II}	& \qw			& \qw \rstick{\ket{1}}		\\
			& \qw		& \qw		& \qw			& \gate[style={yshift=2pt}]{}	& \qw		& \vqw{-3}			& \qw				& \qw				& \qw		& \Upgate{2}{III}	& \qw \rstick{\ket{0}}		\\
			& \qw		& \qw		& \qw			& \qw				& \qw		& \qw				& \gate[style={yshift=2pt}]{}	& \gate[style={yshift=2pt}]{}	& \qw		& \vqw{-4}		& \qw \rstick{\ket{0}}		
		\end{tikzcd}
	};
	\end{tikzpicture}\\[-2ex]
	\caption{Comparison of two circuits mapping Dicke states $\dicke{n}{\ell}$ to computational basis states as a precursor to quantum compression: 
	(left) Mapping to $\ket{0}^{\otimes n-\ell}\ket{1}^{\otimes \ell}$ [this paper].
	(right) Direct mapping to one-hot encoding $\ket{0}^{\otimes \ell-1}\ket{1}\ket{0}^{\otimes n-\ell}$~\cite{Plesch2010}.}
	\label{fig:comparison}
\end{figure}
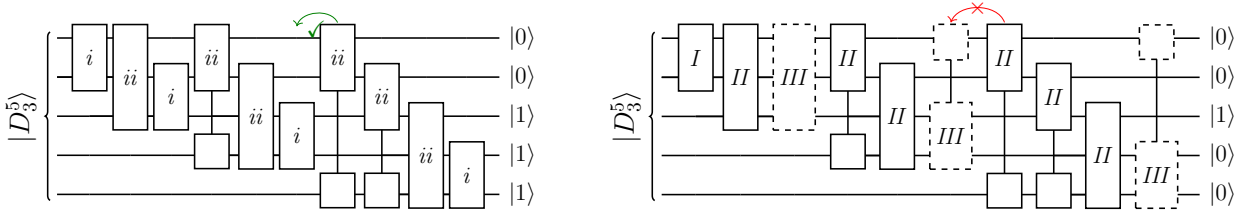

In this appendix, we briefly review the differences between our quantum compression circuit and the one given by Plesch and Bu\v{z}ek~\cite{Plesch2010}.
Both circuits start with a unitary mapping Dicke states to computational basis states, with a small difference: While the latter maps every Dicke state $\dicke{n}{\ell}$
to the one-hot encoding $\ket{0}^{\otimes \ell-1}\ket{1}\ket{0}^{\otimes n-\ell}$ of its Hamming weight $\ell$, 
our circuit for starters maps it to the state $\ket{0}^{\otimes n-\ell}\ket{1}^{\otimes \ell}$. 
Surprisingly, our approach leads to a quadratically improved circuit depth of $\bigO(n)$ over $\bigO(n^2)$.

The two ideas behind the different mappings can be best described like this: 
In several rounds, from $a=2$ to $a=n$, we ``scan'' the first $a$ qubits for excited qubits, ``accumulating'' all found $\ket{1}$s at the bottom of the scanned $a$ qubits. 
In contrast, the existing work \emph{counts} the number of encountered excited qubits, storing the counter in the one-hot encoding in the first $a$ qubits.
These actions are defined in the following gates, see Figure~\ref{fig:comparison} for placements:\\
\begin{minipage}[t]{0.44\textwidth}
	\footnotesize
	\begin{align*}
		\blocki\hspace*{21ex}
		\ket{00}_{a-1}										& \rightarrow \ket{00}_{a-1}			\\
		\ket{11}_{a-1}										& \rightarrow \ket{11}_{a-1}			\\	
		\smash{\sqrt{\tfrac{1}{a}}\ket{01}_{a-1} + \sqrt{\tfrac{a-1}{a}}\ket{10}}_{a-1}		& \rightarrow \ket{01}_{a-1}			\\[3ex]
		\blockii{}\hspace*{19.4ex}
		\ket{00}_b\ket{0}_a									& \rightarrow \ket{00}_b\ket{0}_a		\\
		\ket{01}_b\ket{0}_a									& \rightarrow \ket{01}_b\ket{0}_a		\\
		\ket{00}_b\ket{1}_a									& \rightarrow \ket{00}_b\ket{1}_a		\\
		\ket{11}_b\ket{1}_a									& \rightarrow \ket{11}_b\ket{1}_a		\\[-1ex]
		\sqrt{\tfrac{a-b}{a}} \ket{01}_b\ket{1}_a + \sqrt{\tfrac{b}{a}} \ket{11}_b\ket{0}_a	& \rightarrow \ket{01}_b\ket{1}_a	
	\end{align*}
	\normalsize
\end{minipage}\hfill%
\begin{minipage}[t]{0.52\textwidth}
	\footnotesize
	\begin{align*}
		\mathit{(I)}\hspace*{28ex}
		\ket{00}_1										& \rightarrow \ket{00}_1			\\
		\ket{11}_1										& \rightarrow \ket{01}_1			\\	
		\smash{\sqrt{\tfrac{1}{2}}\ket{01}_1 + \sqrt{\tfrac{1}{2}}\ket{10}}_1			& \rightarrow \ket{10}_1			\\[3ex]
		\mathit{(II)}\hspace*{24ex}
		\ket{00}_b\ket{0}_a									& \rightarrow \ket{00}_b\ket{0}_a		\\
		\ket{10}_b\ket{0}_a									& \rightarrow \ket{10}_b\ket{0}_a		\\
		\ket{00}_b\ket{1}_a									& \rightarrow \ket{00}_b\ket{1}_a		\\
		\ket{01}_b\ket{1}_a									& \rightarrow \ket{01}_b\ket{1}_a		\\[-1ex]
		\sqrt{\tfrac{b+1}{a}} \ket{10}_b\ket{1}_a + \sqrt{\tfrac{a-b-1}{a}} \ket{01}_b\ket{0}_a	& \rightarrow \ket{01}_b\ket{0}_a		\\[3ex]
		\mathit{(III)}\hspace*{21.4ex}
		\ket{0}_1\ket{00}_{a-1}									& \rightarrow \ket{0}_1\ket{00}_{a-1}		\\
		\ket{0}_1\ket{10}_{a-1}									& \rightarrow \ket{0}_1\ket{10}_{a-1}		\\
		\ket{0}_1\ket{11}_{a-1}									& \rightarrow \ket{0}_1\ket{01}_{a-1}		\\[-1ex]
		\sqrt{\tfrac{1}{a}}\ket{0}_1\ket{01}_{a-1}+\sqrt{\tfrac{a-1}{a}}\ket{1}_1\ket{00}_{a-1}	& \rightarrow \ket{1}_1\ket{00}_{a-1}
	\end{align*}
	\normalsize
\end{minipage}\\[2ex]
A comparison of the results after each gate in the two approaches can be found in Table~\ref{tbl:comparison}. 
One can clearly see the two ideas of ``accumulating'' versus ``counting'' the encountered excited qubits: 
After the end of each $a$-round, the states are the same up to a transformation on the first $a$ qubits of the form 
$\ket{0}^{\otimes a-\ell}\ket{1}^{\otimes \ell} \longleftrightarrow \ket{0}^{\otimes \ell-1}\ket{1}\ket{0}^{\otimes a-\ell}$.

\begin{table}[ht!]
	\renewcommand{\arraystretch}{1.3}
	\centering
	\scriptsize
	\begin{tabular}{@{}llll@{}}
		Gate			& $a$ $b$	& Result after each step, this paper (w/o parallelisation)						& Result after each step, existing work~\cite{Plesch2010}						\\
		\hline\hline                                                                                                                                                                                                                                            
		init			&  		& $\sqrt{10}\dicke{5}{3}$										& $\sqrt{10}\dicke{5}{3}$										\\
		\hline                                                                                                                                                                                                                                                  
		$\mathit{i/I}$		& 2		& $\ket{00}\ket{111} + \sqrt{6}\ket{01}\dicke{3}{2} + \sqrt{3}\ket{11}\dicke{3}{1}$			& $\ket{00}\ket{111} + \sqrt{6}\ket{10}\dicke{3}{2} + \sqrt{3}\ket{01}\dicke{3}{1}$			\\
		\hline                                                                                                                                                                                                                                                  
		$\mathit{ii/II}$	& 3 1		& $\ket{001}\ket{11} + \sqrt{2}\ket{010}\ket{11} + \sqrt{6}\ket{011}\dicke{2}{1} + \ket{111}\ket{00}$	& $\ket{001}\ket{11} + \sqrt{2}\ket{100}\ket{11} + \sqrt{6}\ket{010}\dicke{2}{1} + \ket{011}\ket{00}$	\\
		$\mathit{i/III}$	& 3 		& $\sqrt{3}\ket{001}\ket{11} + \sqrt{6}\ket{011}\dicke{2}{1} + \ket{111}\ket{00}$			& $\sqrt{3}\ket{100}\ket{11} + \sqrt{6}\ket{010}\dicke{2}{1} + \ket{001}\ket{00}$			\\
		\hline                                                                                                                                                                                                                                                  
		$\mathit{ii/II}$	& 4 1		& $\sqrt{3}\ket{001}\ket{11} + \sqrt{3}\ket{011}\ket{01} + \sqrt{4}\ket{011}\ket{10}$			& $\sqrt{6}\ket{010}\ket{01} + \sqrt{3}\ket{010}\ket{10} + \ket{001}\ket{00}$				\\
		$\mathit{ii/II}$	& 4 2		& $\sqrt{6}\ket{0011}\ket{1} + \sqrt{4}\ket{0111}\ket{0}$						& $\sqrt{6}\ket{0100}\ket{1} + \sqrt{4}\ket{0010}\ket{0}$						\\
		$\mathit{i/III}$	& 4		& $\sqrt{6}\ket{0011}\ket{1} + \sqrt{4}\ket{0111}\ket{0}$						& $\sqrt{6}\ket{0100}\ket{1} + \sqrt{4}\ket{0010}\ket{0}$ 						\\
		\hline                                                                                                                                                                                                                                                  
		$\mathit{ii/II}$	& 5 1		& $\sqrt{6}\ket{0011}\ket{1} + \sqrt{4}\ket{0111}\ket{0}$						& $\sqrt{6}\ket{0100}\ket{1} + \sqrt{4}\ket{0010}\ket{0}$						\\
		$\mathit{ii/II}$	& 5 2		& $\sqrt{10}\ket{00111}$										& $\sqrt{10}\ket{00100}$										\\	
		$\mathit{ii/II}$	& 5 3		& $\sqrt{10}\ket{00111}$										& $\sqrt{10}\ket{00100}$										\\
		$\mathit{i/III}$	& 5 		& $\sqrt{10}\ket{00111}$										& $\sqrt{10}\ket{00100}$										\\
		\hline
	\end{tabular}
	\normalsize
	\caption{Comparison of the result after each gate given in the two circuits of Figure~\ref{fig:comparison} for the exemplary input state $\dicke{5}{3}$
	(the normalization factor of $\tfrac{1}{\surd10}$ has been omitted in all states).}
	\label{tbl:comparison}
\end{table}

The main difference in the circuits is that for each $a$, our gates form a stair-shape ending in a gate of type $\blocki$, 
while a direct mapping to the one-hot encoding needs to have a ``wrap-around'' three-qubit gate of type $\mathit{(III)}$.
The latter prevents the parallelization achieved by pushing stairs into each other as described in Section~\ref{sec:circuit-size}, 
see Figure~\ref{fig:comparison}.\footnote{We remark that in the original presentation~\cite{Plesch2010} of circuit Figure~\ref{fig:comparison}~(right), 
	it was claimed that gates of type $\mathit{(III)}$ are not necessary to process the Dicke state $\dicke{5}{3}$. 
	However, as can be seen in Table~\ref{tbl:comparison}, the first such gate \emph{is required}, and in general, 
	for a Dicke state $\dicke{n}{\ell}$ gates of type $\mathit{(III)}$ are required up to round $a = n-\ell+1$ 
(the latest position where one may encounter an excited state for the very first time).}
Hence all gates need to be applied sequentially, giving a circuit of depth $\bigO(n^2)$ compared to our depth $\bigO(n)$ construction.

\end{document}